\documentclass[USenglish,11pt]{article}
 
\usepackage{fullpage}
\usepackage{tablefootnote}
\usepackage{algorithm}
\usepackage{algcompatible}
\usepackage{graphics}
\usepackage{setspace}
\usepackage{float}
\usepackage{hyperref}
\hypersetup{colorlinks=true,linkcolor=blue,citecolor=blue}

\usepackage{color}
\usepackage{amsthm}
\usepackage{tikz}
\tikzstyle{vertex}=[circle, draw, inner sep=0pt, minimum size=6pt]

\usepackage{caption}
\usetikzlibrary{patterns,positioning,decorations.pathmorphing,decorations.pathreplacing}
\usetikzlibrary{arrows,automata,snakes}
\usepackage{graphicx}
\usepackage{amsmath}
\usepackage{amssymb}
\usepackage{diagbox}
\usepackage{caption}

\newtheorem{example}{Example}

\newtheorem{corollary}{Corollary}
\newtheorem{lemma}{Lemma}
\newtheorem{claim}{Claim}

\newtheorem{theorem}{Theorem}

\mathchardef\mhyphen="2D
\newcommand\MSA{\mathcal{M}(\mathrm{SA})}
\renewcommand\v{{\texttt{v}}}
\newcommand\CST{\mathrm{SAL}}
\newcommand\TSC{\mathrm{2\mhyphen SAL}}
\newcommand\ST{\mathrm{ST}}
\newcommand\CoST{\mathcal{C}(\mathrm{ST})}
\newcommand\CoMST{\mathcal{C}(\mathrm{MxST})}

\newcommand\MxSTMin{M\footnotesize{xST}\normalsize M\footnotesize{INIMIZATION}\normalsize}
\newcommand\CMxSTMin{CM\footnotesize{xST}\normalsize M\footnotesize{INIMIZATION}\normalsize}
\newcommand\CSTMin{CS\footnotesize{T}\normalsize M\footnotesize{INIMIZATION}\normalsize}
\newcommand\MSAMin{M\footnotesize{SA}\normalsize M\footnotesize{INIMIZATION}\normalsize}

\newcommand\SA{\mathrm{SA}}
\newcommand\MxST{\mathrm{MxST}}
\newcommand\Ls{$L(\sigma)$}
\newcommand\ks{k$_\sigma$}

\title{\textbf{Building Efficient and Compact Data Structures for Simplicial Complexes}\footnote{An extended abstract of this paper appeared in the proceedings of the 31$^{\text{st}}$ International Symposium on Computational Geometry.}}
\date{}
\author{Jean-Daniel Boissonnat\footnote{Geometrica, INRIA Sophia Antipolis - M\'{e}diterran\'{e}e, France. Email: \texttt{Jean-Daniel.Boissonnat@inria.fr}. This work was partially supported by the Advanced Grant of the European Research Council GUDHI (Geometric Understanding in Higher Dimensions). }
\and 
Karthik C.\ S.\footnote{Department of Computer Science and Applied Mathematics, Weizmann Institute of Science, Israel. Email: \texttt{karthik.srikanta@weizmann.ac.il}. This work was partially supported by Irit Dinur's ERC-StG grant number 239985.  Some parts of this work were done at ENS Lyon and at University of Nice - Sophia Antipolis, and were supported by LIP fellowship and Labex UCN@Sophia scholarship respectively.} 
\and
S\'ebastien Tavenas\footnote{ Microsoft Research, India. Email: \texttt{t-sebat@microsoft.com}.
A part of this work was done at LIP, ENS Lyon (UMR 5668 ENS Lyon - CNRS - UCBL - INRIA, Universit\'{e} de Lyon).}}

\begin{document}
\maketitle
\begin{abstract}
\noindent  The Simplex Tree (ST)  is a recently introduced data structure
that can represent abstract simplicial complexes of any dimension and
allows efficient implementation of a large range of basic operations on
simplicial complexes. In this paper, we show how to optimally compress
the Simplex Tree while retaining its functionalities. In addition, we
propose two new data structures called the Maximal Simplex Tree (MxST)  and  the Simplex Array List (SAL). We analyze the compressed
Simplex Tree, the Maximal Simplex Tree, and the Simplex Array List under various settings.
\end{abstract}
\clearpage

\section{Introduction}
Simplicial complexes are widely used in combinatorial and computational topology, and have found many applications in topological data analysis and geometric inference. 
The most common representation uses the Hasse diagram of the complex that
has one node per simplex and an edge between any pair of  incident simplices whose dimensions differ by one. A few attempts to obtain more compact representations have been reported recently.

 Attali et al.\ \cite{DataStructure3} proposed the skeleton-blockers
data structure which represent a simplicial complex by its 1-skeleton
together with its set of blockers. Blockers are the simplices which
are not contained in the complex but whose proper subfaces are. Flag
complexes have no blockers and the skeleton-blocker
representation is  especially efficient for complexes that are
``close'' to flag complexes. An interesting property of the
skeleton-blocker representation is that it  enables efficient edge contraction.

Boissonnat and Maria \cite{SimplexTree} have proposed a tree
representation called the Simplex Tree that can represent general
simplicial complexes and scales well with dimension.  The nodes of the
tree are in bijection with the simplices (of all dimensions) of the
simplicial complex. In this way, the Simplex Tree explicitly stores
all the simplices of the complex but it does not represent explicitly
all the incidences between simplices that are stored in the Hasse
diagram. Storing all the simplices is useful (for example, one can
then attach information to each simplex or store a  filtration
efficiently). Moreover,  the tree structure
of the Simplex Tree leads to efficient implementation of the basic operations on
simplicial complexes (such as retrieving incidence relations, and in
particular retrieving the faces or the cofaces of a simplex).

In this paper, we propose a way to compress the Simplex Tree so as to
store as few nodes and edges as possible without compromising the
functionality of the data structure. The new compressed data
structure is in fact a finite automaton (referred to in this paper as
the Minimal Simplex Automaton) and we describe an optimal algorithm for
its construction. Previous works have looked at trie compression and
have tried to establish a good trade-off between speed and size, but
in most of the works, the emphasis is on one of the two. Two examples
of work where the speed is of main concern are \cite{TrieSpeed1} where
the query time is improved by reducing the number of levels in a
binary trie (which corresponds to truncating the Simplex Tree at a
certain height) and \cite{TrieSpeed2} where trie data structures are
optimized for computer memory architectures. Other popular compact representations for tries in connection with predictive text compression are discussed in \cite{patternMatching}, but they only include (all) substrings of constant length that exist in a text and also do not focus on supporting efficient access in the compressed trie. Therefore, such representations are not useful here, due to the loss of significant information. 

When the size of the structure is of primary concern, the focus is
usually on automata compression. For instance, in the context
of natural language data processing, significant savings in memory
space can be obtained if the dictionary is stored in a directed
acyclic word graph (DAWG), a form of a minimal deterministic
automaton, where common suffixes are shared
\cite{AutomataCompression}. However, theoretical analysis of
compression is seldom done (if at all), in any of these works. In this
paper, we analyze the size of the Minimal Simplex Automaton and also
demonstrate (through experiments) that compression works especially
well for Simplex Tree due to the structure of simplicial complexes: namely, 
that all subfaces of a simplex in the complex also belong
to the complex. Additionally, we consider the influence of the labeling of the
vertices on  compression, which can be significant.  Further, we show that it is hard to find
an optimal labeling for the compressed Simplex Tree
and for the Minimal Simplex Automaton.

We introduce two new data structures for simplicial complexes called
the Maximal Simplex Tree (MxST) and the Simplex Array List (SAL).
MxST is a subtree of the
Simplex Tree whose leaves are in bijection with the maximal
simplices (i.e., simplices with no cofaces) of the complex. We show that this data structure is compact  and that it allows
efficient operations. MxST is
augmented to obtain SAL
where every node uniquely represents an edge. A nice feature of SAL is
its invariance over  labeling of
vertices. We show that SAL  supports efficient basic operations  and
that it is compact when the dimension of the simplicial complex is fixed,
a case of great interest in Manifold Learning and
Topological Data Analysis. 

\section{Simplicial Complex: Definitions and a Lower Bound}\label{S2}
A simplicial complex $K$ is defined over a (finite) vertex set $V$ whose elements are called the vertices of $K$ and is a set of non-empty subsets of $V$ that is required to satisfy the following two conditions:
\begin{enumerate}
\item $p\in V\Rightarrow \{p\}\in K$
\item $\sigma\in K, \tau\subseteq\sigma\Rightarrow\tau\in K$
\end{enumerate}

Each element $\sigma\in K$ is called a simplex or a face of $K$ and, if $\sigma\in K$ has precisely $s + 1$ elements $(s \ge -1)$, $\sigma$ is called an $s$-simplex and the dimension of $\sigma$ is $s$. The dimension of the simplicial complex $K$ is the largest $d$ such that it contains a $d$-simplex.

A face of a simplex $\sigma = \{p_0 ,..., p_s \}$ is a simplex whose vertices form a subset of $\{p_0 ,..., p_s \}$. A proper face is a face different from $\sigma$ and the facets of $\sigma$ are its proper faces of maximal dimension. A simplex $\tau\in K$ admitting $\sigma$ as a face is called a coface of $\sigma$.

In this paper, the class of $d$ dimensional simplicial complexes on
$n$ vertices with $m$ simplices, of which $k$ are maximal, is denoted by
${\cal K}(n,k,d,m)$,  and $K$ denotes a simplicial complex in $ {\cal
  K} (n,k,d,m)$
  .  At times, we say
$K_{\theta}\in {\cal K}_\theta(n,k,d,m)$ (where $\theta:\
V\rightarrow \{1,2,...,|V|\}$ is a labeling of the vertex set $V$ of $K$) when we want to emphasize that some of the data structures seen in this paper are influenced by the labeling of the vertices.

A maximal simplex of a simplicial complex is a simplex which is not contained in a larger simplex of the complex. 
A simplicial complex is pure, if all its maximal simplices are of the
same dimension. Also, a free pair is defined as a pair of simplices
$(\tau,\sigma)$ in $K$ where $\tau$ is the only coface of $\sigma$.  In Figure~\ref{fig:SimplicialComplexExample}, we have a simplicial complex on vertex set $\{1,2,3,4,5,6\}$ which has three maximal simplices: the two tetrahedra 1--3--4--5 and 2--3--4--5, and the triangle 1--3--6. We use this complex as an example through out the paper.

\begin{figure}[!h]
\centering
\resizebox{4cm}{!}{
\begin{tikzpicture}[-,>=stealth',shorten >=0.5pt,auto,node distance=2cm,
 thick,main node/.style={circle,fill=blue!10,draw,font=\sffamily\large\bfseries}]
\draw [fill=blue!10!white] (0,0)--(3,-5)--(6,0) --(8.75,4.75)-- (3,5) --cycle; 
\draw [fill=blue!15!white] (2.5,-1.5)--(3,5) --(0,0)-- (3,-5) --cycle; 
\draw [fill=blue!12!white] (2.5,-1.5)--(3,5) --(0,0) --cycle; 
\draw [fill=blue!8!white] (6,0) --(8.75,4.75)-- (3,5) --cycle; 
\Huge
\draw [dashed] (0,0) -- (6,0);
\draw [-] (0,0) -- (3,5) -- (6,0);
\draw [-] (0,0) -- (2.5,-1.5) -- (6,0);
\draw [-] (3,5) -- (2.5,-1.5) -- (3,-5);
\draw [-] (0,0) -- (3,-5) -- (6,0);
\draw [-] (3,5) -- (8.75,4.75) -- (6,0);

\node at (-0.25,0) {4};
\node at (6.35,-0.1) {3};
\node at (3,5.45) {1};
\node at (3,-5.35) {2};
\node at (9.1,4.75) {6};
\node at (2.9,-1.8) {5};
\end{tikzpicture}
}
\caption{Simplicial complex with the tetrahedra 1--3--4--5 and 2--3--4--5, and the triangle 1--3--6.}
\label{fig:SimplicialComplexExample}
\end{figure}
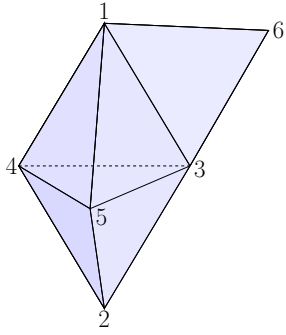

The flag complex of an undirected graph $G$ is defined as an abstract simplicial complex, whose simplices are the sets of vertices in the cliques of $G$. Let $(P, d)$ be a metric space where $P$ is a discrete point set. Given a positive real number $r > 0$, the Rips complex is the abstract simplicial complex $\mathcal{R}^r(P)$ where a simplex  $\sigma\in\mathcal{R}^r(P)$ if and only if $d(p, q) \le r$ for every pair of vertices of $\sigma$. Note that the Rips complex is a special case of a flag complex. This completes the definition of the complexes which will be used in this paper.

We would like to note here that the case when
$k=\mathcal{O}(n)$, is of particular interest. It can be
observed in flag complexes, constructed from planar graphs and
expanders \cite{ELS10}, and in general, from nowhere dense
graphs \cite{GKS13}, and  also from chordal
graphs \cite{G80}. Generalizing, for all flag complexes constructed from graphs with
degeneracy 
$\mathcal{O}(\log n)$
(degeneracy is the smallest integer $r$ such that every subgraph has a vertex of degree at most $r$),
we have that
$k=n^{\mathcal{O}(1)}$ \cite{ELS10}. This encompasses a large class of complexes
encountered in practice. 

Now, we obtain a lower bound on the space needed to represent simplicial
complexes by presenting a counting argument on the number of distinct simplicial complexes.
\begin{theorem}
\label{Lowerbound}
Consider the class of all simplicial complexes on $n$ vertices of dimension $d$, containing $k$ maximal simplices, where $d\ge2$ and $k\ge n+1$, and consider 
any data structure that can represent the simplicial complexes of this class. Such a data structure requires $\log{\binom{\binom{n/2}{d+1}}{k-n}}$ bits to be stored. For any constant  $\varepsilon\in (0,1)$ and for  $\frac{2}{\varepsilon}n\le k\le n^{(1-\varepsilon)d}$ and $d\le n^{\varepsilon/3}$,  the bound becomes $\Omega(kd\log n)$.
\end{theorem}
\begin{proof} The proof of the first statement  is by contradiction. 
  Let us define $h=k-n\ge 1$ and suppose that there exists a data structure that can be stored using only $s<\log \alpha \stackrel{{\rm def}}{=} \log{\binom{\binom{n/2}{d+1}}{h}}$ bits. 
We will construct $\alpha$ simplicial complexes, all with the same set  $P$ of $n$ vertices,  the same dimension $d$, and with exactly $k$ maximal simplices. By the pigeon hole principle, two different simplicial complexes, say $K$ and $K^\prime$, are encoded by the same word. So any algorithm will give the same answer for $K$ and $K^\prime$. But, by the construction of these complexes, there is a simplex which is in $K$ and not in $K^\prime$. This leads to a contradiction. 

The simplicial complexes are constructed as follows. Let $P'\subset P$ be a subset of cardinality $n/2$, and consider the set of all possible simplicial complexes of dimension $d$ with vertices in $P'$ that contain  $h$ maximal simplices. We further assume that all maximal simplices have dimension $d$ exactly. These complexes are $\alpha=\binom{\binom{n/2}{d+1}}{h}$ in number, since the total number of maximal $d$ dimensional simplices is $\binom{n/2}{d+1}$ and we choose $h$ of them. Let us call them $\Gamma_1,\ldots,\Gamma_\alpha$.  We now extend each $\Gamma_i$ so as to obtain a simplicial complex whose vertex set is $P$ and has exactly $k$ maximal simplices. The maximal simplices will consist of  the $h$ maximal simplices of dimension $d$ already constructed plus a number of maximal simplices of dimension 1.  The set of vertices of $\Gamma_i$, ${\rm vert}(\Gamma_i)$,  may be a strict subset of $P'$. Let its cardinality be $\frac{n}{2}-r_i$  and observe that $0\leq r_i<\frac{n}{2}$. Consider now the complete graph on the $\frac{n}{2}+r_i$ vertices of $P\setminus {\rm vert}(\Gamma_i)$.  Any spanning tree of this graph gives $\frac{n}{2}+r_i-1$ edges and we arbitrarily choose $\frac{n}{2}-r_i+1$ edges from the remaining edges of the graph to obtain $n$ distinct edges spanning over the  vertices of $P\setminus {\rm vert}(\Gamma_i)$.  We have thus constructed a 1--dimensional simplicial complex $K_i$ on the $\frac{n}{2}+r_i$ vertices of $P\setminus {\rm vert}(\Gamma_i)$ with exactly $n$ maximal simplices. Finally, we define the complex $\Lambda_i=\Gamma_i\cup K_i$ that has $P$ as its vertex set, dimension $d$, and  $k$ maximal simplices. 
The set of $\Lambda_i$, $i=1, \cdots ,\alpha$, is the set of simplicial complexes we were  looking for.

The second statement in the theorem is proved through the following computation:
\begin{align*}
\log\dbinom{\dbinom{n/2}{d+1}}{k-n}\ge &\ \log{\left(\frac{n^{(d+1)(k-n)}}{2^{(d+1)(k-n)}(d+1)^{(d+1)(k-n)}(k-n)^{(k-n)}}\right)}\\
= &\ (d+1)(k-n)\log n -(d+1)(k-n)\\
&\ \ - (d+1)(k-n)\log (d+1) - (k-n)\log (k-n) \\
> &\ (d+1)(k-n)\log n -3(d+1)(k-n)-(d+1)(k-n)\log d - (k-n)\log k \\
\ge &\ (d+1)(k-n)(\log n -3-\log d) - (k-n)(1-\varepsilon)d\log n \\
\ge &\ d\varepsilon(k-n)\log n + (k-n)\log n -(d+1)(k-n)(3+\frac{\varepsilon}{3}\log n) \\
\ge &\ \frac{2\varepsilon}{3}\left(1-\frac{\varepsilon}{2}\right)kd\log n + (k-n)\log n -3d(k-n)-(k-n)(3+\frac{\varepsilon}{3}\log n) \\
= &\ \Omega(kd\log n)
\end{align*}
We note that in the above computation, the first inequality is obtained by applying the following bound on binomial coefficients: $\binom{n}{d}\ge\left(\frac{n}{d}\right)^d$.
\end{proof}\renewcommand{\thefootnote}{\fnsymbol{footnote}}

We can adapt the above proof to build $n$ maximal
simplices on $P\setminus {\rm vert}(\Gamma_i)$ each of dimension $d$,
to ensure the lower bound applies also to  pure simplicial complexes. This is done by first building $\lfloor\frac{\lvert P\setminus {\rm vert}(\Gamma_i)\rvert}{d+1}\rfloor$ disjoint maximal simplices of dimension $d$ on vertices of $P\setminus {\rm vert}(\Gamma_i)$ and then, building  one maximal simplex which contains all the remaining vertices. We would complete the construction of  $k$ maximal simplices in the complex by choosing $n-\lfloor\frac{\lvert P\setminus {\rm vert}(\Gamma_i)\rvert}{d+1}\rfloor-1$ new maximal simplices of dimension $d$ from vertices of $P\setminus {\rm vert}(\Gamma_i)$. 

Theorem~\ref{Lowerbound} applies particularly to the case of pseudomanifolds of fixed
dimension  where we have $k\le n^\frac{d}{2}$ (i.e., $\varepsilon=\frac{1}{2}$ suffices) \cite{pseudomanifold}. The case where $d$ is small is important in Manifold Learning where it is usually assumed that the data live close to a manifold of small intrinsic dimension. The dimension of the simplicial complex should reflect this fact and ideally be equal to the dimension of the manifold. 

\section{Compression of the Simplex Tree}\label{S3}
Let $K\in {\cal K}(n,k,d,m)$ be a simplicial complex whose  vertices are  labeled from 1 to $n$ and ordered accordingly. We can thus associate to each simplex of $K$ a word on the alphabet set $\{1, \ldots ,n\}$. Specifically, a $j$-simplex of $K$ is uniquely represented as the word of length $j + 1$ consisting of the ordered set of the labels of its $j + 1$ vertices. Formally, let $\sigma = \{v_{\ell_0} , \ldots , v_{\ell_j} \}$ be a simplex, where $v_{\ell_i}$ are vertices of $K$ and $\ell_i \in \{1, \ldots , n\}$ and $\ell_0 <\cdot\cdot\cdot < \ell_j$ . $\sigma$ is represented by the word $[\sigma] = [ \ell_0 , \cdots , \ell_j ]$. The last label of the word representation of a simplex $\sigma$ will be called the last label of $\sigma$ and denoted by last($\sigma$). The simplicial complex $K$ can be defined as a collection of words on an alphabet of size $n$. To compactly represent the set of simplices of $K$, we store the corresponding words in a tree satisfying the following properties:
\begin{enumerate}
\item The nodes of the tree are in bijection with the simplices (of all dimensions) of the complex. The root is associated to the empty face.
\item Each node of the tree, except the root, stores the label of a vertex. Specifically, the node $N$ associated to a simplex $\sigma\neq\emptyset$ stores the label of the vertex last($\sigma$).
\item The vertices whose labels are encountered along a path from the root to a node $N$ associated to a simplex $\sigma$, are the vertices of $\sigma$. The labels are sorted by increasing order along such a path, and each label appears exactly once.
\end{enumerate}

This data structure is called the Simplex Tree of $K$ \cite{SimplexTree} and denoted by ST$(K)$ or simply ST when there is no ambiguity. It may be seen as a trie \cite{StringSearch} on the words representing the simplices of the complex. The depth of the root is 0 and the depth of a node is equal to the dimension of the simplex it represents plus one. Also, in this paper we assume  that ST is directed from the root to the leaves.

We give a constructive definition of ST. Starting from an empty tree,
insert the words representing the simplices of the complex in the
following manner. When inserting the word $[\sigma] = [ \ell_0
,\cdot\cdot\cdot, \ell_j ]$ start from the root, and follow the path
containing successively all labels $\ell_0 , \cdot\cdot\cdot ,
\ell_i$, where $[ \ell_0 ,\cdot\cdot\cdot, \ell_i ]$ denotes the
longest prefix of $[\sigma]$ already stored in the Simplex Tree. Next,
append to the node representing $[ \ell_0 ,\cdot\cdot\cdot, \ell_i ]$
a path consisting of the nodes storing labels $\ell_{i+1}
,\cdot\cdot\cdot, \ell_j$. In Figure 2, we give ST for the simplicial complex shown in Figure~\ref{fig:SimplicialComplexExample}.
\begin{figure}[!h]
\centering
\resizebox{10cm}{!}{
\begin{tikzpicture}[->,>=stealth',shorten >=0.5pt,auto,node distance=2cm,
  thick,main node/.style={circle,fill=blue!10,draw,font=\sffamily\large\bfseries}]
  \node[main node] (0) at (0,0) {X};
  \node[main node] (1) at (-10,-2) {1};
  \node[main node] (2) at (-6,-2) {2};
  \node[main node] (3) at (-2,-2) {3};
  \node[main node] (4) at (0,-2) {4};
  \node[main node] (5) at (2,-2) {5};
  \node[main node] (6) at (4,-2) {6};
  \node[main node] (13) at (-11.5,-4) {3};
  \node[main node] (14) at (-10.5,-4) {4};
  \node[main node] (15) at (-9.5,-4) {5};
  \node[main node] (16) at (-8.5,-4) {6};
  \node[main node] (23) at (-7,-4) {3};
  \node[main node] (24) at (-6,-4) {4};
  \node[main node] (25) at (-5,-4) {5};
  \node[main node] (34) at (-3,-4) {4};
  \node[main node] (35) at (-2,-4) {5};
  \node[main node] (36) at (-1,-4) {6};
  \node[main node] (45) at (0,-4) {5};
  \node[main node] (134) at (-13,-6) {4};
  \node[main node] (135) at (-12,-6) {5};
  \node[main node] (136) at (-11,-6) {6};
  \node[main node] (145) at (-10,-6) {5};
  \node[main node] (234) at (-8,-6) {4};
  \node[main node] (235) at (-7,-6) {5};
  \node[main node] (245) at (-6,-6) {5};
  \node[main node] (345) at (-3,-6) {5};
  \node[main node] (1345) at (-13,-8) {5};
  \node[main node] (2345) at (-8,-8) {5};
      
  \path[every node/.style={font=\sffamily\small}]
    (0) edge node {} (1)
    (0) edge node {} (2)
    (0) edge node {} (3)
    (0) edge node {} (4)
    (0) edge node {} (5)
    (0) edge node {} (6)
    (1) edge node {} (13)
    (1) edge node {} (14)
    (1) edge node {} (15)
    (1) edge node {} (16)
    (2) edge node {} (23)
    (2) edge node {} (24)
    (2) edge node {} (25)
    (3) edge node {} (34)
    (3) edge node {} (35)
    (3) edge node {} (36)
    (4) edge node {} (45)
    (13) edge node {} (134)
    (13) edge node {} (135)
    (13) edge node {} (136)
    (14) edge node {} (145)
    (23) edge node {} (234)
    (23) edge node {} (235)
    (24) edge node {} (245)
    (34) edge node {} (345)
    (134) edge node {} (1345)
    (234) edge node {} (2345)
     ;
     
\end{tikzpicture}
}
\caption{Simplex Tree of the simplicial complex in Figure~\ref{fig:SimplicialComplexExample}.}
\label{fig:SimplexTreeExample}
\end{figure}
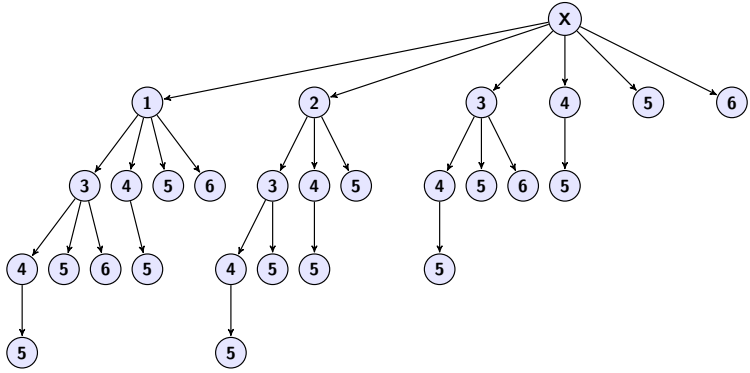

If $K$ consists of $m$ simplices (including the
empty face), the associated ST contains exactly $m$ nodes. Thus, we
need $\Theta(m\log n)$ space/bits to represent ST (since each node stores a vertex which needs $\Theta(\log n)$ bits to be represented). We can compare this
to the lower bound of Theorem~\ref{Lowerbound}. In particular, if
$k=\mathcal{O}(1)$ then, ST requires at least $\Omega(2^d\log n)$ bits
where as Theorem~\ref{Lowerbound} proves the necessity of only
$\Omega(d\log n)$ bits. Therefore, while the Simplex Tree is an efficient
data structure for some basic operations such as determining membership
of a simplex and computing the $r$-skeleton of the complex, it requires storing every simplex explicitly through a node, leading to combinatorial redundancy. To overcome this, we  introduce a compression technique for the ST.

\subsection{Compressed Simplex Tree}\label{S31}

Consider the ST in Figure~\ref{fig:SimplexTreeExampleCommon} and  note that
the red shaded region appears twice. The goal of the compression is
to identify these common parts and store them only once. More concretely, if the
same subtree is rooted at two different nodes in ST then, the subtree
is stored only once and the two root nodes now point to the unique
copy of the subtree. As a consequence, the nodes are no longer in
bijection with the nodes of the complex (as it was in the case of ST),
but we still have the property that the paths from the root are in
bijection with the simplices. We see in
Figure~\ref{fig:CompressedSimplexTreeExample}, the compressed ST of the simplicial complex described in
Figure~\ref{fig:SimplicialComplexExample}. In the rest of the paper,
we denote by $\mathcal{C}$, this action of compression. {\bf Also,
  unless otherwise stated $\mathbf{\lvert \ST\rvert}$ and $\mathbf{\lvert \CoST\rvert}$
  refer to the number of edges in $\mathbf{ST}$ and $\mathbf{\CoST}$ respectively. }

\begin{figure}[!h]
\centering
\resizebox{10cm}{!}{
\begin{tikzpicture}[->,>=stealth',shorten >=0.5pt,auto,node distance=2cm,
  thick,main node/.style={circle,fill=blue!10,draw,font=\sffamily\large\bfseries}]

\filldraw [fill=pink,draw=red] [ultra thick] plot [smooth cycle] coordinates {(-11.5,-3.25) (-10.5,-6) (-13.25,-8.75)  (-13.5,-6)};
\filldraw [fill=pink,draw=red] [ultra thick] plot [smooth cycle] coordinates {(-2,-1.25) (-0.5,-4) (-3.25,-6.75)  (-3.5,-4)};

  \node[main node] (0) at (0,0) {X};
  \node[main node] (1) at (-10,-2) {1};
  \node[main node] (2) at (-6,-2) {2};
  \node[main node] (3) at (-2,-2) {3};
  \node[main node] (4) at (0,-2) {4};
  \node[main node] (5) at (2,-2) {5};
  \node[main node] (6) at (4,-2) {6};
  \node[main node] (13) at (-11.5,-4) {3};
  \node[main node] (14) at (-10.5,-4) {4};
  \node[main node] (15) at (-9.5,-4) {5};
  \node[main node] (16) at (-8.5,-4) {6};
  \node[main node] (23) at (-7,-4) {3};
  \node[main node] (24) at (-6,-4) {4};
  \node[main node] (25) at (-5,-4) {5};
  \node[main node] (34) at (-3,-4) {4};
  \node[main node] (35) at (-2,-4) {5};
  \node[main node] (36) at (-1,-4) {6};
  \node[main node] (45) at (0,-4) {5};
  \node[main node] (134) at (-13,-6) {4};
  \node[main node] (135) at (-12,-6) {5};
  \node[main node] (136) at (-11,-6) {6};
  \node[main node] (145) at (-10,-6) {5};
  \node[main node] (234) at (-8,-6) {4};
  \node[main node] (235) at (-7,-6) {5};
  \node[main node] (245) at (-6,-6) {5};
  \node[main node] (345) at (-3,-6) {5};
  \node[main node] (1345) at (-13,-8) {5};
  \node[main node] (2345) at (-8,-8) {5};
      
  \path[every node/.style={font=\sffamily\small}]
    (0) edge node {} (1)
    (0) edge node {} (2)
    (0) edge node {} (3)
    (0) edge node {} (4)
    (0) edge node {} (5)
    (0) edge node {} (6)
    (1) edge node {} (13)
    (1) edge node {} (14)
    (1) edge node {} (15)
    (1) edge node {} (16)
    (2) edge node {} (23)
    (2) edge node {} (24)
    (2) edge node {} (25)
    (3) edge node {} (34)
    (3) edge node {} (35)
    (3) edge node {} (36)
    (4) edge node {} (45)
    (13) edge node {} (134)
    (13) edge node {} (135)
    (13) edge node {} (136)
    (14) edge node {} (145)
    (23) edge node {} (234)
    (23) edge node {} (235)
    (24) edge node {} (245)
    (34) edge node {} (345)
    (134) edge node {} (1345)
    (234) edge node {} (2345)
     ;
     
\end{tikzpicture}
}
\caption{Common subtrees of the Simplex Tree in Figure~\ref{fig:SimplexTreeExample}.}
\label{fig:SimplexTreeExampleCommon}
\end{figure}
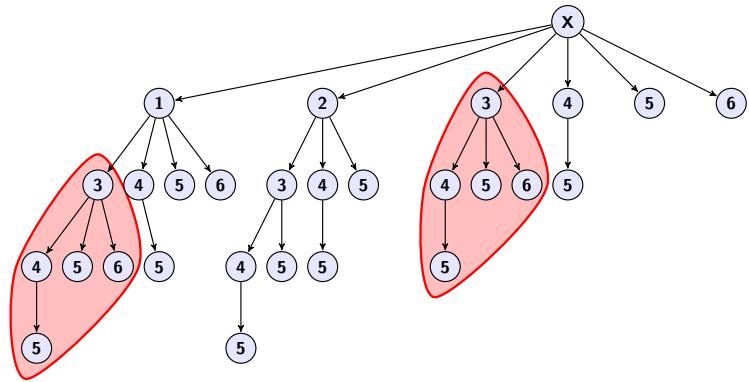

Answering simplex membership queries and other queries that only require traversing ST from root to leaves can be implemented in $\CoST$ exactly as in  ST \cite{SimplexTree}. Allowing upward traversal in ST is also possible (with additional pointers from children to parents), and this has been shown to improve the efficiency of some operations, such as face or coface retrieval. However, in $\CoST$, parents are not unique. To account for this, we mark the parents that were accessed, and use this to go back in the upward direction. This implies an additional storage of $\mathcal{O}(d\log n)$ while traversing, but a node (simplex) having many parents can assist to locate cofaces much faster.

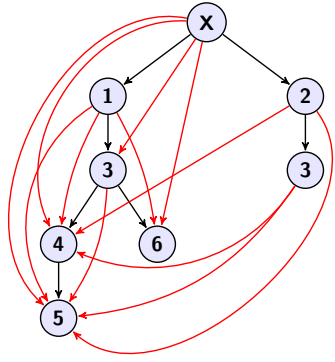
\begin{figure}[!h]
 \centering
 \resizebox{6cm}{!}{
 \begin{tikzpicture}[->,>=stealth',shorten >=0.5pt,auto,node distance=2cm,
  thick,main node/.style={circle,fill=blue!10,draw,font=\sffamily\large\bfseries}]
  \node[main node] (0) at (0,0) {X};
  \node[main node] (1) at (-2,-1.5) {1};
  \node[main node] (2) at (2,-1.5) {2};
  \node[main node] (13) at (-2,-3) {3};
  \node[main node] (23) at (2,-3) {3};
  \node[main node] (134) at (-3,-4.5) {4};
  \node[main node] (136) at (-1,-4.5) {6};
  \node[main node] (1345) at (-3,-6) {5};
      
  \path[every node/.style={font=\sffamily\small}]
    (0) edge node {} (1)
    (0) edge node {} (2)
    (1) edge node {} (13)
    (2) edge node {} (23)
    (13) edge node {} (134)
    (13) edge node {} (136)
    (134) edge node {} (1345)
    (0) edge [red] node {} (13)
    (0) edge [red, bend right = 60] node {} (134)
    (0) edge [red] node {} (136)
    (0) edge [red, bend right = 75] node {} (1345)
    (1) edge [red, bend right = 10] node {} (134)
    (1) edge [red, bend left = 10] node {} (136)
    (1) edge [red, bend right = 45] node {} (1345)
    (2) edge [red, bend left = 85] node {} (1345)
    (2) edge [red] node {} (134)
    (13) edge [red, bend left = 15] node {} (1345)
    (23) edge [red,bend left = 25] node {} (1345)    
    (23) edge [red, bend left = 40] node {} (134)  
     ;
     
\end{tikzpicture}}\caption{Compressed Simplex Tree of the Simplex Tree given in Figure~\ref{fig:SimplexTreeExample}. The extent of compression is demonstrated by the following: the edge $4-5$ which appears six times in the Simplex Tree of Figure~\ref{fig:SimplexTreeExample}, appears only once in the Compressed Simplex Tree}
\label{fig:CompressedSimplexTreeExample}
\end{figure}

Next, we will introduce an automaton perspective of the above compression and show how to  deduce the optimal compression algorithm for ST. We will also describe insertion and removal operations on $\CoST$ through the automaton perspective.

\subsection{Minimal Simplex Automaton}\label{S4}

A Deterministic Finite state Automaton (DFA) recognizing a language is defined by a set of states and labeled transitions between these states to detect if a given word is in a predefined language or not. $\ST$ can be seen as a DFA: let us define the set of $m$ states by $\mathcal{V}=\{\textrm{nodes of }\ST\}$. A transition from state $u$ to state $v$ is labeled by $a$ if and only if there is in $\ST$ an edge from $u$ to $v$, and $v$ contains the vertex $a$. We define the Simplex Automaton of $K$ (denoted by SA($K$)) as the automaton described above (cf. Figure~\ref{fig:SimplexAutomatonExample}).

\begin{figure}[!h]
\centering
\resizebox{12cm}{!}{
\begin{tikzpicture}[->,>=stealth',shorten >=0.5pt,auto,node distance=2cm,
  thick,main node/.style={circle,fill=blue!10,draw,font=\sffamily\large\bfseries}]

  \node[main node] (0) at (0,0) {};
  \node[main node] (1) at (-10,-2) {};
  \node[main node] (2) at (-6,-2) {};
  \node[main node] (3) at (-2,-2) {};
  \node[main node] (4) at (0,-2) {};
  \node[main node] (5) at (2,-2) {};
  \node[main node] (6) at (4,-2) {};
  \node[main node] (13) at (-11.5,-4) {};
  \node[main node] (14) at (-10.5,-4) {};
  \node[main node] (15) at (-9.5,-4) {};
  \node[main node] (16) at (-8.5,-4) {};
  \node[main node] (23) at (-7,-4) {};
  \node[main node] (24) at (-6,-4) {};
  \node[main node] (25) at (-5,-4) {};
  \node[main node] (34) at (-3,-4) {};
  \node[main node] (35) at (-2,-4) {};
  \node[main node] (36) at (-1,-4) {};
  \node[main node] (45) at (0,-4) {};
  \node[main node] (134) at (-13,-6) {};
  \node[main node] (135) at (-12,-6) {};
  \node[main node] (136) at (-11,-6) {};
  \node[main node] (145) at (-10,-6) {};
  \node[main node] (234) at (-8,-6) {};
  \node[main node] (235) at (-7,-6) {};
  \node[main node] (245) at (-6,-6) {};
  \node[main node] (345) at (-3,-6) {};
  \node[main node] (1345) at (-13,-8) {};
  \node[main node] (2345) at (-8,-8) {};
      
  \path[every node/.style={}]
    (0) edge [above] node {1} (1)
    (0) edge node {2} (2)
    (0) edge node {3} (3)
    (0) edge [] node {4} (4)
    (0) edge [below] node {5} (5)
    (0) edge node {6} (6)
    (1) edge [above] node [pos=0.8] {3} (13)
    (1) edge node [pos=0.6] {4} (14)
    (1) edge node [pos=0.8] {5} (15)
    (1) edge node {6} (16)
    (2) edge [above] node {3} (23)
    (2) edge node {4} (24)
    (2) edge node {5} (25)
    (3) edge [above] node {4} (34)
    (3) edge node {5} (35)
    (3) edge node {6} (36)
    (4) edge node {5} (45)
    (13) edge [above] node {4} (134)
    (13) edge node {5} (135)
    (13) edge node {6} (136)
    (14) edge node {5} (145)
    (23) edge [above] node {4} (234)
    (23) edge node {5} (235)
    (24) edge node {5} (245)
    (34) edge node {5} (345)
    (134) edge node {5} (1345)
    (234) edge node {5} (2345)
     ;
     
\end{tikzpicture}
}
\caption{Simplex Automaton of the simplicial complex in Figure~\ref{fig:SimplicialComplexExample}.}
\label{fig:SimplexAutomatonExample}
\end{figure}
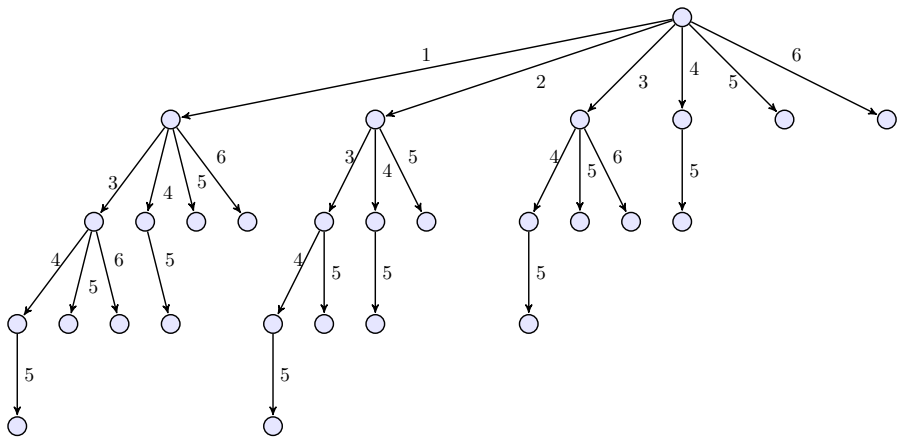

SA is basically  the same data structure as ST except that the labels are not put on the nodes but on the edges entering these nodes and thus, basic operations  in SA can be implemented as in ST. Also, by construction of SA, it is obvious that the number of states and transitions of SA are equal to the number of nodes and edges in ST respectively. 

It is known~\cite{AutomateMinimization2} that if a language $L$ is regular (accepted by a DFA) then,  $L$ has a unique minimal automaton.
DFA minimization is the task of transforming a given DFA into an equivalent DFA that has a minimum number of states. We represent the action of performing DFA minimization by $\mathcal{M}$. For any $K\in {\cal K}_\theta(n,k,d,m)$, let us define the
Minimal Simplex Automaton ($\mathcal{M}(\mathrm{SA})$) as the minimal
deterministic automaton which recognizes the language
$L_{K_\theta}$.
Compressing ST can be seen as DFA minimization since merging
identical subtrees corresponds to merging indistinguishable states in
the automaton. 
It is possible to get $\CoST$ from $\MSA$ by duplicating the states such that for each node, the labels of all its incoming edges are the same, and then by moving the labels from the edges to the next node.  
Also, 
it should be observed that the number of edges in $\CoST$ and the
number of transitions in $\MSA$ may not be the same.  The reason is that states in $\MSA$ having identical set of
outgoing paths can merge even when the incoming set of transitions are
different for each of these states, while such nodes in $\CoST$ would
not have merged. 

\begin{sloppypar}Algorithmic aspects of DFA minimization have been well studied. For instance,
Hopcroft's algorithm~\cite{HopcroftAlgorithm} minimizes an automaton
with $m$ transitions over an alphabet of size $n$ in
$\mathcal{O}(m\log m\log n)$ steps and needs at most
$\mathcal{O}(m\log n)$ space. This running time is shown
in~\cite{HopcroftAlgorithm} to be optimal over the set of regular
languages. Additionally, Revuz showed that acyclic automaton (which SA indeed is) can be minimized in linear time \cite{R92}. Also, in Appendix~\ref{Hopcroft} we describe an adapted Hopcroft's algorithm to optimally compress ST. In Figure~\ref{fig:SimplexMinimalAutomatonExample}, we
give the minimal automaton for the simplicial complex of
Figure~\ref{fig:SimplicialComplexExample}.\end{sloppypar}
\begin{figure}[!h]
\centering
\resizebox{7cm}{!}{
\begin{tikzpicture}[->,>=stealth',shorten >=0.5pt,auto,node distance=2cm,
  thick,main node/.style={circle,fill=blue!10,minimum size = 25pt,draw,font=\sffamily\large\bfseries}]
  \node[main node] (0) at (0,0) {};
  \node[main node] (1) at (-4,-1.5) {};
  \node[main node] (2) at (4,-1.5) {};
  \node[main node] (13) at (-2,-3) {};
  \node[main node] (23) at (2,-3) {};  
  \node[main node] (134) at (0,-4.5) {};
  \node[main node] (1345) at (0,-6) {};

  \path[every node/.style={}]
    (0) edge [above] node {1} (1)
    (0) edge node {2} (2)
    (0) edge [above] node {3} (13)
    (0) edge [left] node {4} (134)
    (0) edge [bend left = 23] node [pos = 0.3] {5,6} (1345)
    (1) edge node {3} (13)
    (1) edge [left, bend right = 30] node {4} (134)
    (1) edge [left, bend right = 40] node {5,6} (1345)
    (2) edge [above] node {3} (23)
    (2) edge [bend left = 30] node {4} (134)
    (2) edge [bend left = 40] node {5} (1345)
    (13) edge node {4} (134)
    (23) edge [left] node [pos = 0.2] {4} (134)    
    (13) edge [left, bend right = 15] node [pos = 0.7] {5,6} (1345)
    (134) [left] edge node {5} (1345)
     ;
     
\end{tikzpicture}
}
\caption{Minimal Simplex Automaton of the simplicial complex in Figure~\ref{fig:SimplicialComplexExample}.}
\label{fig:SimplexMinimalAutomatonExample}
\end{figure}
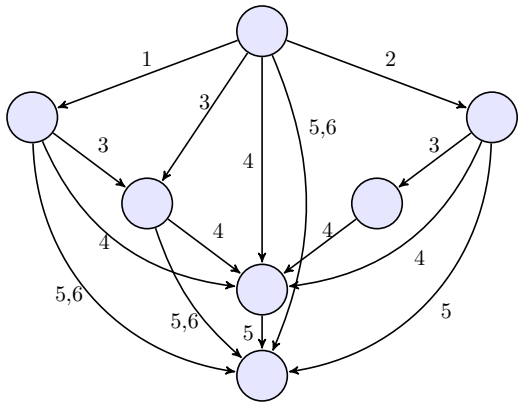

While there are delicate differences between
$\MSA$ and $\CoST$, we will see below that performing basic operations on $\MSA$ is not very different from the way it is done for $\CoST$.

\subsubsection{Operations on the Minimal Simplex Automaton}
\label{op_MSA}

The set of all paths originating from the root are the same in both ST and $\mathcal{M}(\mathrm{SA})$. All operations which involve only traversal along ST are performed with equal (if not better) efficiency in $\mathcal{M}(\mathrm{SA})$ as, for every such operation on ST, we start by traversing from the root. As an example, consider the operation of determining if a simplex $\sigma$ is in the complex. Let us adapt the algorithm described in \cite{SimplexTree} to $\mathcal{M}(\mathrm{SA})$. Note that there is a unique path from the initial state which identifies $\sigma$ in $\MSA$. If $\sigma=v_{\ell_0}-\cdots - v_{\ell_{d_\sigma}}$ then, from the initial state we go through $d_\sigma+1$ states by following the transitions $\ell_0,\ldots ,\ell_{d_\sigma}$ in that order. If at some point the requisite transition is not found then, declare that the simplex is not in the complex. Hence, performing all static operations i.e., all operations where we don't change the $\MSA$ in any way,  can be carried out in very much the same way in both  $\mathcal{M}(\mathrm{SA})$ and ST, although it might be more efficient for $\MSA$ as discussed earlier for $\CoST$ in subsection~\ref{S31}.

Addition and deletion of simplexes can be trickier in
$\mathcal{M}(\mathrm{SA})$ than in ST. We can always expand
$\mathcal{M}(\mathrm{SA})$ to SA, (locally) perform the operation and
recompress. If the nature of the operation is itself expensive
(i.e., worst-case $\Omega(m)$) then, the worst-case cost does not change, which is indeed the case for operations such as removal of cofaces, edge contraction and elementary collapses.

\subsubsection{Complexity Measure of Size  for the Minimal Simplex Automaton}
Our complexity measure in the paper would be minimizing SA to obtain
$\MSA$ with minimum number of states. We know from Myhill-Nerode
theorem that there is a unique minimal DFA. Thus, if there was an 
automaton  with less transitions than $\MSA$ then, it should 
have more states than $\MSA$. Let $A$ be an automaton with a minimal 
number of transitions and let us run 
Hopcroft's algorithm on $A$. Since, at no point in Hopcroft's 
algorithm, we increase the number of transitions,  the output has to be an automaton whose numbers  
of states and transitions are both minimal. It follows that the unique automaton output 
by Hopcroft's algorithm must have both a minimal number of states and a minimal number of transitions. This is proved more formally as Proposition 1 in \cite{MMR13}. 

In the rest of the  paper, \textbf{we denote by $\mathbf{\lvert \SA\rvert}$ and $\mathbf{\lvert\MSA\rvert}$, the number of states in $\mathbf{\SA}$ and 
$\mathbf{\MSA}$ respectively}. 
While we consider the number of states as a complexity measure of the 
size of 
$\SA$ and $\MSA$, we will still use 
the number of edges as a complexity measure of the size of  $\ST$ 
and $\CoST$ because the bounds obtained relating the number of edges in $\CoST$ and the number of nodes in $\CoST$ are not satisfactory.
The size of $\MSA$ will be discussed in detail in section~\ref{S6}, after introducing a new data structure in the next section. This is done to put the impact of compression in better perspective.

\section{Maximal Simplex Tree}\label{S5}
We define the {\em Maximal Simplex Tree} MxST$(K)$ 
as an induced subgraph of ST($ K$). All leaves in the Simplex Tree corresponding to maximal simplices and the nodes encountered  on the path from the root to these leaves are kept in the Maximal Simplex Tree and the remaining nodes are removed. 
MxST$(K)$ is constructed as follows. We start from an empty tree and then insert the words representing the maximal simplices of $K$. Specifically, when inserting the word $[\sigma] = [ \ell_0 ,\cdot\cdot\cdot, \ell_j ]$, we start from the root, and follow the path containing successively all labels $\ell_0 , \cdot\cdot\cdot , \ell_i$, where $[ \ell_0 ,\cdot\cdot\cdot, \ell_i ]$ denotes the longest prefix of $[\sigma]$ already stored in the Maximal Simplex Tree. We then append to the node representing $[ \ell_0 ,\cdot\cdot\cdot, \ell_i ]$ a path consisting of the nodes storing labels $\ell_{i+1} ,\cdot\cdot\cdot, \ell_j$.  
Figure~\ref{fig:MaximalSimplexTreeExample}  shows the MxST of the simplicial complex given in Figure~\ref{fig:SimplicialComplexExample}. In $\MxST(K)$, the leaves are in bijection with the maximal simplices of $K$.  Any path starting from the root provides the vertices of a simplex of $K$. However, in general, not all simplices in $K$ can be associated to a path from the root in $\MxST (K)$.
\begin{figure}[!h]
\centering
\resizebox{4cm}{!}{
 
\begin{tikzpicture}[->,>=stealth',shorten >=0.5pt,auto,node distance=2cm,
  thick,main node/.style={circle,fill=blue!10,draw,font=\sffamily\large\bfseries}]
  \node[main node] (0) at (0,0) {X};
  \node[main node] (1) at (-2,-1.5) {1};
  \node[main node] (2) at (2,-1.5) {2};
  \node[main node] (13) at (-2,-3) {3};
  \node[main node] (23) at (2,-3) {3};
  \node[main node] (134) at (-3,-4.5) {4};
  \node[main node] (136) at (-1,-4.5) {6};
  \node[main node] (234) at (2,-4.5) {4};
  \node[main node] (1345) at (-3,-6) {5};
  \node[main node] (2345) at (2,-6) {5};
      
  \path[every node/.style={font=\sffamily\small}]
    (0) edge node {} (1)
    (0) edge node {} (2)
    (1) edge node {} (13)
    (2) edge node {} (23)
    (13) edge node {} (134)
    (13) edge node {} (136)
    (23) edge node {} (234)
    (134) edge node {} (1345)
    (234) edge node {} (2345)
     ;

\end{tikzpicture}
}
\caption{Simplicial Complex of Figure~\ref{fig:SimplicialComplexExample} represented using Maximal Simplex Tree.}
  \label{fig:MaximalSimplexTreeExample}
\end{figure}
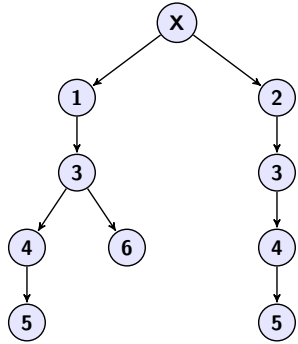

By the above construction of MxST, we add at most $d+1$ nodes per maximal simplex. Hence, MxST$(K)$ has at most $k(d+1)+1$ nodes and at most $k(d+1)$ edges (therefore requiring $\mathcal{O}(kd\log n)$ space). \textbf{We denote by $\mathbf{\lvert \MxST\rvert}$ the number of edges in $\mathbf{\MxST}$.} Since MxST is a factor of ST, the size of MxST is usually much smaller than the size of ST. Further, it always meets the lower bound of Theorem~\ref{Lowerbound}, making it a compact data structure. We discuss below the efficiency of MxST in answering queries.

\subsection{Operations on the Maximal Simplex Tree}

In \cite{SimplexTree} some important basic operations (with appropriate motivation) have been discussed for ST. We will bound now the cost of these operations using MxST. Note that any node in MxST$(K)$ has $\mathcal{O}(n)$ children and  we can search  for a particular child of a node in time $\mathcal{O}(\log n)$ (using red--black trees). We summarize in Table~\ref{tab:OperationsofMxST}, the asymptotic cost of some basic operations (details of this analysis is provided in Appendix~\ref{MxSTOperations}) and note that it is already better than ST for some operations. 

\begin{table}[!h]
\begin{center}
\begin{tabular}{|p{7.4cm}|p{2.9cm}|}\hline
\multicolumn{1}{|c|}{Operation}&\multicolumn{1}{|c|}{Cost}\\\hline
Identifying maximal cofaces of a simplex $\sigma$ / Determining membership of $\sigma$ &$\mathcal{O}(kd\log n)$  \\\hline
Insertion of a maximal simplex $\sigma$ & $\mathcal{O}(kd_{\sigma}\log n)$\\\hline
Removal of a face & $\mathcal{O}(kd\log n)$ \\\hline
Elementary Collapse & $\mathcal{O}(kdd_\sigma\log n)$ \\\hline
Edge Contraction & $\mathcal{O}(kd(k+\log n))$
\\\hline
\end{tabular}
\end{center}
\caption{Cost of performing basic operations on MxST.}
  \label{tab:OperationsofMxST}
\end{table}

Moreover, we can augment the structure of MxST without paying for a lot of extra memory space, so that the above operations can be performed more efficiently. This is explained in section~\ref{S7}. 

\section{Results on  Minimization of the Simplex Automaton}\label{S6}
In this section we will see some results, both theoretical and experimental on the minimization of SA.

\subsection{Bounds on the Number of States of the Minimal Simplex Automaton}

We observe below that the number of leaves in ST is large and grows linearly w.r.t.\ the number of nodes in ST. The proof follows by a simple induction argument on $n$.
\begin{lemma}
  \label{Intuitionof CST1}
If $K\in {\cal K}(n,k,d,m)$ then, at least half the nodes of $\ST(K)$ are leaves.
\end{lemma}
\begin{proof}
The proof is by induction on the number  $n$ of vertices. When $n=1$, we have two simplices (including the empty simplex) and one leaf. Now assume (induction hypothesis) that for all simplicial complexes on $i$ vertices the corresponding ST has at least half of its nodes as leaves. Consider a ST on the vertex set $\{1,2,...,i+1\}$ containing $m$ simplices. Consider the subtree under node 1 say $\ST_1$. $\ST_1$ represents a simplicial complex on the vertex set $\{2,3,...,i+1\}$ with node 1 acting as the root. Suppose $\ST_1$ has $m_1$ nodes. Then, by the induction hypothesis, $\ST_1$ has at least $\frac{m_1}{2}$ leaves. Now consider the rest of $\ST$ which can also be independently seen as a Simplex Tree $\ST_2$ on the vertex set $\{2,...,i+1\}$. Again, by the induction hypothesis, $\ST_2$ has at least $\frac{m-m_1}{2}$ leaves. Thus $\ST$ has at least $\frac{m}{2}$ leaves. 
\end{proof}

Differently from ST, $\mathcal{M}(\mathrm{SA})$ has only one leaf. The following lemma shows that $\mathcal{M}(\mathrm{SA})$ has at most half the number of nodes of ST plus one (follows directly from Lemma~\ref{Intuitionof CST1}).
\begin{lemma}\label{statesBound}
  For any $K\in {\cal K} (n,k,d,m)$, $\mathcal{M}(\mathrm{SA}(K))$ has at most $\frac{m}{2}+1$ states.
\end{lemma}
\begin{proof}
  Since there are $m$ nodes in ST, SA has exactly $m$ states and at
  least $\frac{m}{2}$ of them are leaves (a state without outgoing
  transitions). The $\frac{m}{2}$ leaves can be merged. Consequently,
  the number of states of $\mathcal{M}(SA(K))$ is at most $m-\frac{m}{2}+1=\frac{m}{2}+1$.
\end{proof}

Similar to $\MSA$, we may define $\mathcal{M}(\mathrm{MxSA}(K))$ as the minimal DFA which recognizes only maximal simplices as words. Then, the following inequality follows:
\begin{lemma}\label{MxSAstates}
  For any \textbf{pure} simplicial complex $K\in {\cal K} (n,k,d,m)$,
  $\lvert\mathcal{M}(\mathrm{SA}(K))\rvert 
  \ge\lvert\mathcal{M}(\mathrm{MxSA}(K))\rvert$. 
\end{lemma}
\begin{proof}
 Let $K$ be a pure simplicial complex. Let $s_i$ be the initial state and $S_f$ be the set of the states of
  outdegree zero. To each state we associate its depth, which is the
  the length of the longest directed path from $s_i$ to this state.

  We define another automaton $B$. The states of $B$ are the states of
  ${\SA(K)}$. The state $s_i$ is still the initial state and, the new final states are 
    $\{s\in S_f\mid \textrm{depth}(s)=d+1\}$. Finally, there is a transition between
  states $u$ and $v$ labeled by $a$ if and only if this transition
  exists in ${\SA(K)}$ and if $\textrm{depth}(v)=\textrm{depth}(u)+1$.
  Let us prove that $B$ recognizes exactly the maximal simplices of
  $K$. 

  Let $w$ be a word recognized by $B$. Then $w$ is a word of length
  $d+1$ which was recognized by ${\SA(K)}$. Then, it corresponds to a simplex
  of $K$ of dimension $d$. It is a maximal simplex of $K$.

  On the other direction, let $\sigma$ be a maximal simplex of $K$. Hence
  the word $\sigma$ is accepted by ${\SA(K)}$ and of length $d+1$. If any transition which appears
  during this detection appears also in $B$ then, $\sigma$ is also accepted
  by $B$. Thus, let us assume for a contradiction that during the detection of the word
  $\sigma$, the simplex automaton ${\SA(K)}$ uses a transition from a state $u$ to a state
  $v$ such that $\textrm{depth}(v)\neq \textrm{depth}(u)+1$ (let us
  choose the first transition where it happens). By
  definition of the depth of a state, we get $\delta_1=\textrm{depth}(v) >
  \textrm{depth}(u)+1=\delta_2+1$. This means that there exists a word $w$ of length
  $\delta$ such that by reading $w$ from the initial state, we arrive
  into $v$. Then, $w,\sigma_s$ (where $\sigma_s$ is the suffix of $\sigma$ of length
  $d-\delta_2$) is also accepted by ${\SA(K)}$. Consequently $K$
  contains a simplex of dimension $d+\delta_1-\delta_2-1>d$ and we have reached a contradiction.
\end{proof}

In fact, one can prove that for a large class of simplices the
equality does not hold. For instance, consider ${\cal
  K^\prime}\subset {\cal K} (n,k,d,m)$ such that for any $K\in{\cal
  K^\prime}$ we have that there exists two maximal simplices which
have different first letters (i.e., when the simplices are treated as
words) but have the same letter at position $i$, for some $i$ 
that  is not the last position. For this subclass the equality does
not hold. Observe also that Lemma~\ref{MxSAstates} holds only for pure
simplicial complexes because, if all complexes were allowed then, we
will have complexes like in Example~\ref{notPure} where
$\lvert\mathcal{M}(\mathrm{SA}(K))\rvert 
  <\lvert\mathcal{M}(\mathrm{MxSA}(K))\rvert $. 
  
  \begin{example}\label{notPure}
  Consider the simplicial complex  on seven vertices given by the following maximal simplices: a 4-cell 1--2--3--6--7, two triangles 2--3--5 and 4--6--7 and an edge 4--5.
  \end{example}

\subsection{Conditions for Compression}
We would like to analyze two possible sources of compression in ST. A first type of compression may happen when a simplex $\sigma$ belongs to several maximal simplices and its vertices appear as the last vertices of those maximal simplices.
Then compression will factorize $\sigma$ so that it will appear only once as a common suffix of several words.  A second type of compression occurs when considering a single maximal simplex. Here too, ST stores many different words with common suffixes and compression will factorize these common suffixes.
  Intuitively, the first type captures compression solely in MxST 
(i.e., because of the input and the labeling on vertices we have 
defined) and the second source analyzes possible compression because 
of the rich structure of ST.   
Now, we will see a result which guarantees compression for pure simplicial complexes regardless of
the  labeling of the vertices:
\begin{lemma}\label{boundOnCompression}
For any \textbf{pure} simplicial complex $K\in {\cal K} (n,k,d,m)$, we have that $\lvert \MSA\rvert$ is always less than $\lvert \SA\rvert$ when $k<d$ and $d \geq 2$. 
\end{lemma}
\begin{proof}
Let $K$ be a pure simplicial complex. Let $\sigma=v_1,\ldots,v_d,v_{d+1}$ be a maximal simplex. Let us define
  $\nu =\{v_1,\ldots,v_{d-1}\}$ and, 
 $
    M=\{m\in K \mid m\textrm{ is
      maximal and }v_d\in m \}.
 $
  We also define
  $P_\nu\subseteq \mathcal{P}(\nu)$ as the projection of $M$ on to $\nu$, which is more formally written as:
  $$
    P_\nu=\{a\subseteq \nu \mid \exists m\in M, a=m\cap \nu\}.
  $$
  We notice that $\nu \in P_\nu$.
  Since $d > k\geq \lvert M\rvert \geq
  \lvert P_\nu\rvert$, it follows that there exists $b\subseteq \nu$
  such that $\lvert b\rvert =\lvert \nu\rvert -1$ and
  which is not in $P_\nu$. Let $s_\nu$ and $s_b$ be the states in ${\SA(K)}$
  reached by reading the words $\nu\cup v_d$ and $b \cup v_d$. As the language is
  closed by subwords and as $b\subseteq \nu$, any
  accepting word from the state $s_\nu$ is also an accepting word from
  the state $s_b$. Reciprocally, if $w$ is an accepting word from the
  state $s_b$ then, $b\cup v_d\cup w$ is a face of a maximal simplex
  $m$. The projection of $m$ on $\nu$ contains $b$ and by definition
  of $b$ is strictly larger. Hence $\nu \subseteq m$, and so, $\nu\cup
  v_d \cup w\in K$. Consequently the states $s_\nu$ and $s_b$ are
  equivalent and can be merged.
\end{proof}

In fact, the above result is close to tight:  in Example~\ref{NFAvsDFA},
we have a pure simplicial complex with  $k=d$ and $\lvert \ST\rvert=\lvert \mathcal{C}(\ST)\rvert$. We remark here that in cases of impossibility of 
compression we will analyze the compression of ST rather than the minimization of SA through out this section because analyzing ST provides better insight into the combinatorial structures which hinder compression.

Intuitively, it seems natural that if the given simplicial complex has a large number of maximal simplices then, regardless of the labeling we should be able to compress some pairs of nodes in MxST. However, Example~\ref{noExternalCompression} says otherwise.
\begin{example}\label{noExternalCompression}Consider the simplicial complex on $2n$ vertices of dimension $n/2$ defined by the set of maximal simplices given by: 
$$\left\{g(i)\cup \left\{g^r(i)+n\right\}\bigg| i\in \left\{1,2,\dots ,\binom{n}{n/2}\right\} ,r\in\left\{1,2,\dots ,n/2\right\}\right\}$$ 
where $g$ is a bijective map from $\{1,2,\dots ,\binom{n}{n/2}\}$ to
the set of all simplices on $n$ vertices of dimension $n/2 -1$ and $g^r$ corresponds to picking the $r^{\textrm{th}}$ vertex (in lexicographic order). 
\end{example} 
Here $k=\frac{n}{2}\binom{n}{n/2}\approx 2^{n-\frac{1}{2}}\sqrt{n/\pi}$ and there is no compression in MxST. Also note that $\lvert \CoMST\rvert<\lvert \MxST\rvert$ does not imply $\lvert \CoST\rvert<\lvert \ST\rvert$ as can be seen in Example~\ref{MxSTnotST}.

\begin{example}\label{MxSTnotST}
Consider the simplicial complex on seven vertices given by the maximal simplices: tetrahedron 1-2-4-6 and three triangles 2-4-5, 3-4-5 and 1-4-7.
\end{example}

In both Examples~\ref{noExternalCompression}~and~\ref{NFAvsDFA}, we saw simplicial complexes of large
dimension which cannot be compressed, but this is due to the way the vertices were labeled. Now, we state a lemma which says that there is always a labeling which ensures compression.
\begin{lemma}\label{AlwaysCompression}
If $K_\theta\in\mathcal{K}(n,k,d,m)$ with $d>1$ then, we can find a permutation $\pi$ on $\{1,2,\dots ,n\}$ such that $\lvert \mathcal{M}(\SA(K_{\pi\circ\theta}))\rvert<\lvert \SA(K_{\pi\circ\theta})\rvert$.
\end{lemma}
\begin{proof}
Let $m=v_{\ell_0}\cdots v_{\ell_d}$ be a maximal simplex in $K_\theta$. We construct $\pi$ by swapping $\ell_{d-1}$ and $\ell_d$ with $n-1$ and $n$ respectively. In $\ST(K_{\pi\circ\theta})$, note that the node under root with label $n-1$ and the node corresponding to simplex $v_{\pi(\ell_0)}\cdots v_{\pi(\ell_{d-1})}$ are identical and thus can be merged.
\end{proof}


We would have liked to obtain better bounds for the size of $\mathcal{M}(\mathrm{SA})$ through conditions just based on $n,k,d$ and $m$, but sadly this is a hard combinatorial problem. Also, while there is always a good labeling, we show in section~\ref{S8} that it is NP-Hard to find it.

\subsection{Experiments}\label{Exp}
We define two parameters here, $\rho_{_{\ST}}$ and
$\rho_{_{\MxST}}$. The first one is given by the ratio of $|\ST|$
and $|\CoST|$ and the second by the ratio of $|\MxST|$ and
$|\CoMST|$. All experiments performed below record the extent of
compression of ST. Ideally, we would have liked to record the extent
of minimization of SA since $\MSA$ is more compact than $\CoST$. Unfortunately, this has not been possible due to the lack of available libraries able to handle very large automata. The results below for $\CoST$ are nonetheless positive, substantiating our claim that compression of ST leads to a compact data structure.   

\noindent\textbf{Data Set 1:} The set of points  were obtained through
sampling of a Klein bottle in $\mathbb{R}^5$ and constructing the Rips
Complex with parameter $r$ using libraries provided by the GUDHI project \cite{GUDHI} on input of various values for $r$. We record in Table~\ref{tab:DataSet1},  $|\CoST|$ and $|\CoMST|$ for the various complexes constructed.

\begin{table}[!h]
\begin{center}\resizebox{\linewidth}{!}{
\begin{tabular}{|c|c|c|c|c|c|c|c|c|c|c|c}\hline
No&$n$&$r$&$d$&$k$&$\lvert\ST\rvert=m-1$&$\lvert \MxST\rvert$&$\lvert\CoST\rvert$&$\rho_{_{\ST}}$&$\lvert\CoMST\rvert$&$\rho_{_{\MxST}}$
\\\hline
1&10,000&0.15&10&24,970&604,572&96,104&218,452&2.77&90,716&1.06
\\\hline
2&10,000&0.16&13&25,410&1,387,022&110,976&292,974&4.73&104,810&1.06
\\\hline
3&10,000&0.17&15&27,086&3,543,582&131,777&400,426&8.85&123,154&1.07
\\\hline
4&10,000&0.18&17&27,286&10,508,485&149,310&524,730&20.03&137,962&1.08
\\\hline
\end{tabular}}
\end{center}
\caption{Analysis of experiments on Data Set 1.}
\label{tab:DataSet1}
\end{table}

First, observe that $|\MxST|$ is considerably smaller than $|\ST|$. This is expected, as it is likely that $k$ is polynomially related to $n$ for Rips complexes. Also, while we observe insignificant compression in  MxST, $\rho_{_\ST}$ increases rapidly as $r$ is increased. This indicates that compression strongly exploits the 
combinatorial redundancy of ST (i.e., storing each simplex explicitly 
through a node)  
and 
works particularly well for the  Simplex Tree.  

\noindent\textbf{Data Set 2:}  All experiments conducted above
are for Rips complexes with  $\frac{d}{n}$  small. We now check the
extent of compression for  simplicial complexes with large
$\frac{d}{n}$. To this aim, we look at flag complexes generated using a random graph $G_{n,p}$ on $n$ vertices where a pair of vertices share an edge with probability $p$, and record in Table~\ref{tab:DataSet2},  $|\CoST|$ and $|\CoMST|$ for the various complexes constructed.

\begin{table}[!h]
\begin{center}\resizebox{\linewidth}{!}{
\begin{tabular}{|c|c|c|c|c|c|c|c|c|c|c|c|c}\hline
No&$n$&$p$&$d$&$k$&$\lvert\ST\rvert=m-1$&$\lvert \MxST\rvert$&$\lvert\CoST\rvert$&$\rho_{_{\ST}}$&$\lvert\CoMST\rvert$&$\rho_{_{\MxST}}$
\\\hline
1&25&0.8&17&77&315,369&587&467&537.3&121&4.85
\\\hline
2&30&0.75&18&83&4,438,558&869&627&7,079.0&134&6.49
\\\hline
3&35&0.7&17&181&3,841,590&1,592&779&4,931.4&245&6.50
\\\hline
4&40&0.6&19&204&9,471,219&1,940&896&10,570.6&276&7.03
\\\hline
5&50&0.5&20&306&25,784,503&2,628&1,163&22,170.7&397&6.62
\\\hline
\end{tabular}}
\end{center}
\caption{Analysis of experiments on Data Set 2.}
\label{tab:DataSet2}
\end{table}

Here we observe staggering values for $\rho_{_\ST}$ which increases as the simplicial complex grows larger. This is primarily because random simplicial complexes don't behave like pathological simplicial complexes (such as Examples~\ref{noExternalCompression}~and~\ref{NFAvsDFA}) which hinder compression. 

\section{Simplex Array List}\label{S7}
In this section, we build a new data structure which is a hybrid of ST and MxST.
The {\em Simplex Array List} $\CST(K)$ is a (rooted) directed acyclic graph on at most $k\left(\frac{d(d+1)}{2}+1\right)$ nodes with maximum out-degree $d$, which can be obtained by modifying MxST or  can be  constructed from the maximal simplices of $K$. We describe the construction of SAL below.

\subsection{Construction}
 We will first see how to obtain $\CST$ from MxST by performing three operations which we define below. 
\begin{enumerate}
\item \textbf{Unprefixing $(\mathcal{U})$:} Excluding the root and the leaves, for every node $v$ in $\MxST$ with outdegree $d_v$, duplicate it into $d_v$ nodes with outdegree 1, (one copy of $v$ for each of its children) by starting from the parents of the leaves and recursively moving up in the tree. 
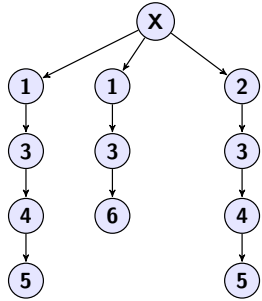
\begin{figure}[!h]
\centering
\resizebox{3.5cm}{!}{
 
\begin{tikzpicture}[->,>=stealth',shorten >=0.5pt,auto,node distance=2cm,
  thick,main node/.style={circle,fill=blue!10,draw,font=\sffamily\Large\bfseries}]
  \node[main node] (0) at (0,0) {X};

  \node[main node] (1y) at (-3,-1.5) {1}; 
\node[main node] (1x) at (-1,-1.5) {1};

  \node[main node] (2) at (2,-1.5) {2};
  \node[main node] (13x) at (-1,-3) {3};
  \node[main node] (13y) at (-3,-3) {3};  
  \node[main node] (23) at (2,-3) {3};
  \node[main node] (134) at (-3,-4.5) {4};
  \node[main node] (136) at (-1,-4.5) {6};
  \node[main node] (234) at (2,-4.5) {4};
  \node[main node] (1345) at (-3,-6) {5};
  \node[main node] (2345) at (2,-6) {5};

  \path[every node/.style={font=\sffamily\small}]    (1y) edge node {} (13y)     (0) edge node {} (1y)       ;
      
  \path[every node/.style={font=\sffamily\small}]
    (0) edge node {} (1x)

    (0) edge node {} (2)
(1x) edge node {} (13x)
    (2) edge node {} (23)
    (13y) edge node {} (134)
    (13x) edge node {} (136)
    (23) edge node {} (234)
    (134) edge node {} (1345)
    (234) edge node {} (2345)
     ;

\end{tikzpicture}
}\caption{Unprefixing the Maximal Simplex Tree of Figure~\ref{fig:MaximalSimplexTreeExample}.}
  \label{fig:UnprefixedMaximalSimplexTreeExample}
\end{figure}

\item \textbf{Transitive Closure $(\mathcal{T})$:} For every pair of nodes $(u,v)$ in $\mathcal{U}(\MxST)$ ($u$ not being the root), if there is a path from $u$ to $v$ then, add an edge from $u$ to $v$ in $\mathcal{T}(\mathcal{U}(\MxST))$ (if it doesn't already exist).

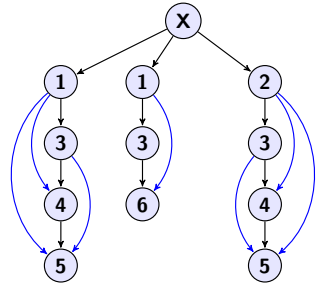
\begin{figure}[!h]
\centering
\resizebox{4.5cm}{!}{
 
\begin{tikzpicture}[->,>=stealth',shorten >=0.5pt,auto,node distance=2cm,
  thick,main node/.style={circle,fill=blue!10,draw,font=\sffamily\Large\bfseries}]
  \node[main node] (0) at (0,0) {X};

   \node[main node] (1y) at (-3,-1.5) {1}; 
\node[main node] (1x) at (-1,-1.5) {1};

  \node[main node] (2) at (2,-1.5) {2};
  \node[main node] (13x) at (-1,-3) {3};
  \node[main node] (13y) at (-3,-3) {3};  
  \node[main node] (23) at (2,-3) {3};
  \node[main node] (134) at (-3,-4.5) {4};
  \node[main node] (136) at (-1,-4.5) {6};
  \node[main node] (234) at (2,-4.5) {4};
  \node[main node] (1345) at (-3,-6) {5};
  \node[main node] (2345) at (2,-6) {5};

  \path[every node/.style={font=\sffamily\small}]    (1y) edge node {} (13y)     (0) edge node {} (1y)       ;
      
  \path[every node/.style={font=\sffamily\small}]
    (0) edge node {} (1x)

    (0) edge node {} (2)
(1x) edge node {} (13x)
    (2) edge node {} (23)
    (13y) edge node {} (134)
    (13x) edge node {} (136)
    (23) edge node {} (234)
    (134) edge node {} (1345)
    (234) edge node {} (2345)
     ;

 \path[every node/.style={font=\sffamily\small}]
    (1y) edge [blue, bend right=40] node {} (134)
    (1y) edge [blue, bend right=50] node {} (1345)    

;
  \path[every node/.style={font=\sffamily\small}]
    (1x) edge [blue, bend left=40] node {} (136)
    (13y) edge [blue, bend left=40] node {} (1345)    
    (2) edge [blue, bend left=40] node {} (234)
    (2) edge [blue, bend left=50] node {} (2345)    
    (23) edge [blue, bend right=40] node {} (2345)    
     ;

\end{tikzpicture}
}\caption{Transitive Close of the Unprefixed Maximal Simplex Tree of Figure~\ref{fig:UnprefixedMaximalSimplexTreeExample}.}
  \label{fig:TransitiveMaximalSimplexTreeExample}
\end{figure}

\item \textbf{Expanding Representation$(\mathcal{R})$:} For every node $v$ in $\mathcal{T}(\mathcal{U}(\MxST))$ with outdegree $d_v$, duplicate it into $d_v$ nodes with outdegree 1, i.e., one copy of $v$ for each of its children, by starting from the children of the root and recursively moving down to children of smallest label. As a demonstration, in Figure~\ref{fig:ExpandNodeMaximalSimplexTreeExample}, we show the result of expanding \emph{one} node in the $\mathcal{T}(\mathcal{U}(\MxST)$ of Figure~\ref{fig:TransitiveMaximalSimplexTreeExample}.

\begin{figure}[!h]
\centering
\resizebox{5.3cm}{!}{
 
\begin{tikzpicture}[->,>=stealth',shorten >=0.5pt,auto,node distance=2cm,
  thick,main node/.style={circle,fill=blue!10,draw,font=\sffamily\Large\bfseries}]
  \node[main node] (0) at (0,0) {X};

\node[main node] (1x) at (-1,-1.5) {1};

  \node[main node] (2) at (2,-1.5) {2};
  \node[main node] (13x) at (-1,-3) {3};
  \node[main node] (13y) at (-3,-3) {3};  
  \node[main node] (23) at (2,-3) {3};
  \node[main node] (134) at (-3,-4.5) {4};
  \node[main node] (136) at (-1,-4.5) {6};
  \node[main node] (234) at (2,-4.5) {4};
  \node[main node] (1345) at (-3,-6) {5};
  \node[main node] (2345) at (2,-6) {5};

  \path[every node/.style={font=\sffamily\small}]
    (0) edge node {} (1x)

    (0) edge node {} (2)
(1x) edge node {} (13x)
    (2) edge node {} (23)
    (13y) edge node {} (134)
    (13x) edge node {} (136)
    (23) edge node {} (234)
    (134) edge node {} (1345)
    (234) edge node {} (2345)
     ;

\node[main node] (1y4) at (-4,-1.5) {1};
\node[main node] (1y3) at (-2.5,-1.5) {1};
\node[main node] (1y5) at (-5.5,-1.5) {1};
\path[every node/.style={font=\sffamily\small}]
    (1y4) edge [blue, bend right=10] node {} (134)
    (1y5) edge [blue] node {} (1345)    
 (1y3) edge node {} (13y)     (0) edge node {} (1y3) (0) edge node {} (1y4)(0) edge [bend right=10] node {} (1y5)      ;
;

  \path[every node/.style={font=\sffamily\small}]
    (1x) edge [blue, bend left=40] node {} (136)
    (13y) edge [blue, bend left=40] node {} (1345)    
    (2) edge [blue, bend left=40] node {} (234)
    (2) edge [blue, bend left=50] node {} (2345)    
    (23) edge [blue, bend right=40] node {} (2345)    
     ;

\end{tikzpicture}
}\caption{Expanding \textbf{one} node in $\mathcal{T}(\mathcal{U}(\MxST)$ of Figure~\ref{fig:TransitiveMaximalSimplexTreeExample}.}
  \label{fig:ExpandNodeMaximalSimplexTreeExample}
\end{figure}
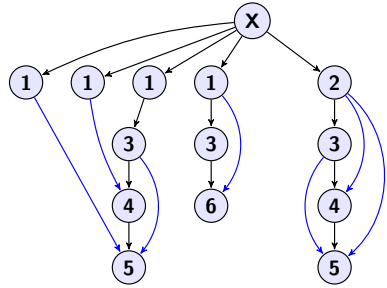
\end{enumerate}    

$\mathcal{R}(\mathcal{T}(\mathcal{U}(\MxST)))$ is the Simplex Array List. From the construction of $\CST$, it is clear that each node in $\CST$ uniquely represents an edge in the simplicial complex. Figure~\ref{fig:CompactSimplexTreeExample}  shows SAL representation of the simplicial complex given in Figure~\ref{fig:SimplicialComplexExample}.   

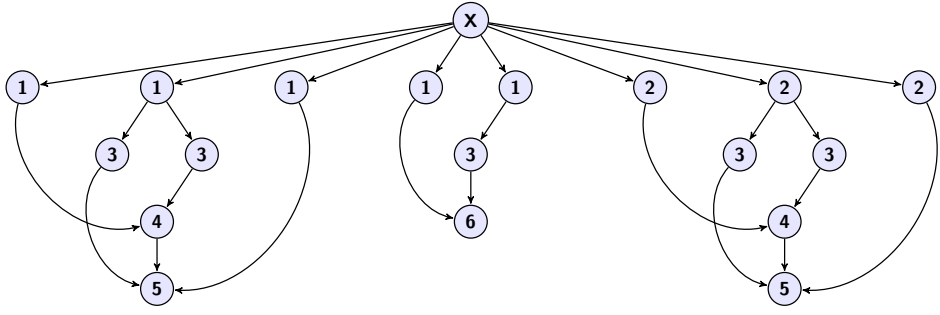
\begin{figure}[!h]
\centering
\resizebox{13cm}{!}{
 
\begin{tikzpicture}[->,>=stealth',shorten >=0.5pt,auto,node distance=2cm,
  thick,main node/.style={circle,fill=blue!10,draw,font=\sffamily\large\bfseries}]

       \node[main node] (1) at (16,0) {X};
  \node[main node] (2) at (6,-1.5) {1};
  \node[main node] (3) at (9,-1.5) {1};
  \node[main node] (45) at (12,-1.5) {1};      
  \node[main node] (5) at (15,-1.5) {1};  
  \node[main node] (6) at (17,-1.5) {1};    
  \node[main node] (7) at (20,-1.5) {2};
  \node[main node] (8) at (23,-1.5) {2};  
  \node[main node] (90) at (26,-1.5) {2};  
  \node[main node] (10) at (8,-3) {3};
  \node[main node] (11) at (10,-3) {3};  
  \node[main node] (12) at (16,-3) {3};  
  \node[main node] (13) at (22,-3) {3};
  \node[main node] (14) at (24,-3) {3};    
  \node[main node] (15) at (9,-4.5) {4};
  \node[main node] (16) at (16,-4.5) {6};
  \node[main node] (17) at (23,-4.5) {4};
  \node[main node] (18) at (9,-6) {5};
  \node[main node] (19) at (23,-6) {5};
      
  \path[every node/.style={font=\sffamily\small}]
    (1) edge node {} (2)
    (1) edge node {} (3)    
    (1) edge node {} (5)
    (1) edge node {} (6)
    (1) edge node {} (7)
    (1) edge node {} (8)
    (1) edge node {} (45)
    (1) edge node {} (90)    
    (2) edge [bend right = 60] node {} (15)
    (45) edge [bend left = 60] node {} (18)   
    (3) edge node {} (10)
    (3) edge node {} (11)
    (11) edge node {} (15)    
    (15) edge node {} (18)   
    (10) edge [bend right = 60] node {} (18)     
    (7) edge [bend right = 60] node {} (17)
    (90) edge [bend left = 60] node {} (19)   
    (8) edge node {} (13)
    (8) edge node {} (14)
    (14) edge node {} (17)    
    (17) edge node {} (19)   
    (13) edge [bend right = 60] node {} (19)         
    (5) edge [bend right = 60] node {} (16)     
    (6) edge node {} (12)    
    (12) edge node {} (16)               
    
     ;
     
\end{tikzpicture}
}
\caption{Simplicial Complex of Figure~\ref{fig:SimplicialComplexExample} represented using  $\mathcal{R}(\mathcal{T}(\mathcal{U}(\MxST)))$.}

  \label{fig:CompactSimplexTreeExample}

\end{figure}

We will now see an equivalent construction of $\CST$ from its maximal simplices and it is this construction we will use to perform operations. For a given maximal simplex $\sigma=v_{\ell_0} \cdot\cdot\cdot v_{\ell_j}$, associate a unique key between $1$ and $k$ generated using a hash function $\mathcal{H}$ and then introduce $\frac{j(j+1)}{2}+1$ new nodes in $\CST$. We build a set of $\frac{j(j+1)}{2}+1$ labels and assign uniquely a label to each node. The set of labels is defined as the union of the following two sets
 (cf.\ Figure~\ref{fig:CompactSimplexTreeExampleNumberLine} for an example):\\
$S_1=\{(\ell_{i},\ell_{i^\prime},\mathcal{H}(\sigma))\mid i\in\{0,1,\dots ,j-1\},i^\prime\in\{i+1,\dots ,j\}\}$\\
$S_2=\{(\ell_j,\varphi,\mathcal{H}(\sigma))\}$\\
where $\varphi$ denotes an empty label. We introduce an edge from node with label $(\ell_{p},\ell_{p^\prime},\mathcal{H}(\sigma))$ to node with label $(\ell_{q},\ell_{q^\prime},\mathcal{H}(\sigma))$ if and only if $p^\prime=q$. Additionally, we introduce an edge from every node with label $(\ell_{p},\ell_j,\mathcal{H}(\sigma))$ in $S_1$ to the node with label  $(\ell_j,\varphi,\mathcal{H}(\sigma))$ in $S_2$. Thus, in $\CST$ we represent a maximal $j$-simplex using a connected component containing $|S_1|+|S_2|=\frac{j(j+1)}{2}+1$ nodes and $\frac{j(j^2+5)}{6}$ directed edges. To perform basic operations efficiently, we embed $\CST$ on the number line such that for every $i\in\{1,2,\dots ,n\}$, we have an array $A_i$ of nodes which has labels of the form $(i,i^\prime,z)$ for some $z\in\{1,\dots ,k\}$ and $i^\prime\in\{i+1,\dots ,n,\varphi\}$. Sort each $A_i$ based on $i^\prime$ and in case of ties, sort them based on $z$.

The resultant graph obtained after removing the root in $\mathcal{R}(\mathcal{T}(\mathcal{U}(\MxST)))$ is the same as the one described in the previous paragraph. Labels (as described above) for the nodes in $\mathcal{R}(\mathcal{T}(\mathcal{U}(\MxST)))$ can be easily given by just looking at the vertex represented by the node, and its children. 

We remark here that  we use hash function $\mathcal{H}$ to generate keys for simplices because it is an efficient way to reuse keys (in case of multiple insertions and removals).

\begin{figure}[!h]
\centering
\resizebox{12cm}{!}{
 
\begin{tikzpicture}[->,>=stealth',shorten >=0.5pt,auto,node distance=2cm,
  thick,main node/.style={circle,fill=blue!10,draw,font=\sffamily\large\bfseries}]
  \node[main node] (113) at (1,1) {};
  \node[main node] (133) at (1,2) {};  
  \node[main node] (114) at (1,3) {};
  \node[main node] (115) at (1,4) {};    
  \node[main node] (136) at (1,5) {};    
  \node[main node] (223) at (3,4) {};  
  \node[main node] (224) at (3,5) {};
  \node[main node] (225) at (3,6) {};  
  \node[main node] (314) at (5,2) {};
  \node[main node] (324) at (5,3) {};  
  \node[main node] (315) at (5,4) {};  
  \node[main node] (325) at (5,5) {};
  \node[main node] (336) at (5,6) {};    
  \node[main node] (415) at (7,3) {};
  \node[main node] (425) at (7,4) {};
  \node[main node] (51p) at (9,3) {};
  \node[main node] (52p) at (9,4) {};
  \node[main node] (63p) at (11,4) {};
      
  \path[every node/.style={font=\sffamily}]
    (113) edge [red!50!yellow, bend left=15] (314)
    (113) edge [red!50!yellow] (315)
    (114) edge [red!50!yellow, bend right=75] (415)    
    (115) edge [red!50!yellow, bend right =65] (51p)
    (314) edge [red!50!yellow] (415)
    (315) edge [red!50!yellow] (51p)
    (415) edge [red!50!yellow] (51p)
    (223) edge [blue!40!white] (324)
    (223) edge [blue!40!white] (325)
    (224) edge [blue!40!white, bend left=65] (425)    
    (225) edge [blue!40!white, bend left =50] (52p)
    (324) edge [blue!40!white] (425)
    (325) edge [blue!40!white,bend left =25] (52p)
    (425) edge [blue!40!white] (52p)
    (133) edge [green!40!yellow, bend left=10] (336)    
    (136) edge [green!40!yellow,bend left =90] (63p)
    (336) edge [green!40!yellow, bend left =30] (63p)
                    
    
     ;
     \draw [->] (-1.5,0) -- (12,0);
     \fill[black] (-1,0) circle (0.1);
     \fill[black] (1,0) circle (0.1);     
     \fill[black] (3,0) circle (0.1);
     \fill[black] (5,0) circle (0.1);     
     \fill[black] (7,0) circle (0.1);
     \fill[black] (9,0) circle (0.1);     
     \fill[black] (11,0) circle (0.1);
     \node at (-1,-0.35) {0};
     \node at (1,-0.35) {1};
     \node at (3,-0.35) {2};
     \node at (5,-0.35) {3};
     \node at (7,-0.35) {4};
     \node at (9,-0.35) {5};
     \node at (11,-0.35) {6};
     \draw[red!70!black,dashed] (3.7,6.5) rectangle (2.3,3.3);
     \draw[red!70!black,dashed] (1.7,5.5) rectangle (0.3,0.3);     
     \draw[red!70!black,dashed] (4.3,1.3) rectangle (5.7,6.5);
     \draw[red!70!black,dashed] (6.3,2.3) rectangle (7.7,4.5);     
     \draw[red!70!black,dashed] (8.3,2.3) rectangle (9.7,4.5);
     \draw[red!70!black,dashed] (10.3,3.3) rectangle (11.7,4.5);     
	 \node at (1,0.65) {(1,3,1)};
     \node at (1,2.65) {(1,4,1)};          
     \node at (1,3.65) {(1,5,1)};
     \node at (1,1.65) {(1,3,3)};
     \node at (1,4.65) {(1,6,3)};          
	 \node at (3,3.65) {(2,3,2)};
     \node at (3,4.65) {(2,4,2)};          
     \node at (3,5.65) {(2,5,2)};
     \node at (5,1.65) {(3,4,1)};          
	 \node at (5,3.65) {(3,5,1)};
     \node at (5,2.65) {(3,4,2)};
     \node at (5,4.65) {(3,5,2)};          
     \node at (5,5.65) {(3,6,3)};
     \node at (7,2.65) {(4,5,1)};
     \node at (7,3.65) {(4,5,2)};          
     \node at (9,2.65) {(5,$\varphi$,1)};
     \node at (9,3.65) {(5,$\varphi$,2)};          
     \node at (11,3.65) {(6,$\varphi$,3)};
	\large{\node at (1, 5.75) {$\mathbf{A_1}$};     }
	\large{\node at (3, 6.75) {$\mathbf{A_2}$};     }
	\large{\node at (5, 6.75) {$\mathbf{A_3}$};     }
	\large{\node at (7, 4.75) {$\mathbf{A_4}$};     }
	\large{\node at (9, 4.75) {$\mathbf{A_5}$};     }
	\large{\node at (11, 4.75) {$\mathbf{A_6}$};     }		
				
\end{tikzpicture}
	}
\caption{Simplex Array List for complex in Figure~\ref{fig:SimplicialComplexExample} embedded on the number line.}

  \label{fig:CompactSimplexTreeExampleNumberLine}

\end{figure}
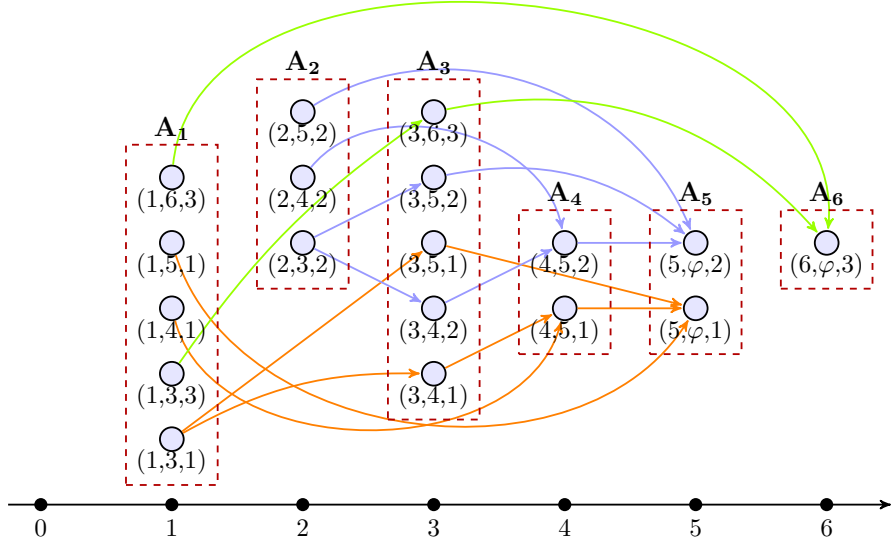

\subsection{Some Observations about the Simplex Array List}

$\CST(K)$ has at most $k\left(\frac{d(d+1)}{2}+1\right)$ nodes. Also, for each maximal simplex of dimension $d_\sigma$, the outdegree of any node in the connected component corresponding to the maximal simplex, is at most $d_\sigma$. Therefore, the total number of edges in $\CST(K)$ is at most $k\left(\frac{d^2(d+1)}{2}+d\right)$. Further, in each node we store the labels of two vertices (which requires $\log n$ bits) and a hashed value (which requires $\log k$ bits). Hence, the space required to store $\CST(K)$ is $\mathcal{O}(kd^2(d+\log n+\log k))$. \textbf{Also, unless otherwise stated $\lvert \CST\rvert$ refers to number of edges in $\CST$.} Since $\CST$  is constructed from $\mathcal{U}(\MxST)$, we have the following lemma:

\begin{lemma}\label{CSTinvariant}
The number of nodes and edges in SAL are both invariant over the labeling of the vertices in the simplicial complex.
\end{lemma}

Intuitively, SAL is representing $K$ by storing all the edges of $K$ explicitly as nodes in $\CST(K)$ and the edges in $\CST(K)$ are used to capture the incidence relations between simplices. More precisely, a path of length $j$ in $\CST(K)$ corresponds to a unique $j$-simplex in $K$. We now see that, differently from MxST, the simplices of $K$ are all associated with paths in $\CST(K)$. We say a path $p$ is associated to a simplex $\sigma$ if the sequence of numbers obtained by looking at the corresponding nodes which are embedded on the number line along $p$ are exactly the labels of the vertices of $\sigma$ in lexicographic order.
\begin{lemma}\label{LemBij}
Any path in $\CST(K)$ is associated to a simplex of  $K$ and  any simplex of $K$ is associated to at least one such path.
\end{lemma}
\begin{proof}
For any path in $\CST(K$), it belongs to a connected component of $\CST(K$). In particular, there exists a maximal simplex $m$ such that all nodes in this component are of the form $(a,b,\mathcal{H}(m))$ for  some $a,b$ such that $v_a$ is a vertex of $m$. Consequently, all vertices read during this path belong to $m$, it means that the corresponding simplex is a face of $m$, so it belongs to $K$.

On the other hand, if $\sigma=v_{\ell_1}\cdots v_{\ell_r}$ is a simplex of $K$, it belongs to some maximal simplex $m=v_{\ell_1^\prime}\cdots v_{\ell_{d_m}^\prime}$ where we have $\{\ell_1,\ell_2,\dots ,\ell_r\}\subseteq\{\ell_1^\prime,\dots ,\ell_{d_m}^\prime\}$. In $\CST(K$), we look at the connected component associated with $m$. While we can associate any node with label $(a,b,z)$ uniquely to edge $a-b$, we can also uniquely identify to vertex $a$. We associate $\sigma$ to the path:
$$(\ell_1,\ell_2,\mathcal{H}(m))-(\ell_2,\ell_3,\mathcal{H}(m))-\cdots -(\ell_r,x,\mathcal{H}(m)),$$
for some $x\in\{\ell_r+1,\ldots ,n,\varphi\}$. The existence of the path is confirmed because of the way we introduced edges in the connected component. 
\end{proof}

Observe that several paths can provide the same simplex since a simplex may appear in several maximal simplices. Hence, the vertices of a given simplex cannot be accessed in a deterministic way. The previous lemma together with this observation implies that  $\CST$ is a non-deterministic finite automaton (NFA). NFA are a natural generalization of DFA. The size of a NFA is smaller than that of a DFA detecting the same language, but the operations on NFA take  in general more time. We demonstrate the above fact using Example~\ref{NFAvsDFA}.
\begin{example}\label{NFAvsDFA}
Let $K\in \mathcal{K}(2k+1,k,k,m)$ be defined on the vertices $\{1,\ldots,2k+1\}$ and the set of maximal simplices be given by $\{(\{1,\ldots,k+1\}\setminus\{i\})\cup\{k+1+i\}\mid1\leq i\leq k\}$. 
\end{example}

Thus $\CST(K)$ has $\frac{k^2(k+1)}{2}+k$ nodes while  $\mathcal{M}(\mathrm{SA}(K))$
has at least $2^k$ states (all states reached after reading the
words $s\subseteq \{1,\ldots,k\}$ are pairwise distinct). Moreover, this motivates the need for considering $\CST$ over $\MSA$, as the gap in their sizes can be exponential.

Building $\CST(K)$ can be seen as partially compressing the Simplex Tree ST$({\sigma})$ associated to each maximal simplex $\sigma$ (where $\sigma$ and its subfaces are seen as a subcomplex). Compressing  ST$({\sigma})$ will lead to a subtree which is exactly the same as the transitive closure of MxST($\sigma$). Therefore, collecting all $\CoST({\sigma})$ for all maximal simplices $\sigma$ and merging the roots is the same as $\mathcal{T}(\mathcal{U}(\MxST(K)))$. Now applying $\mathcal{R}$ on $\mathcal{T}(\mathcal{U}(\MxST(K)))$ can be seen as an act of uncompression. We apply $\mathcal{R}$ once to ensure that for every node, all its children represent the same vertex and thus belong to the same $A_i$. If $\mathcal{R}$ is applied multiple times then, it is equivalent to duplicating nodes (seen as an act of uncompression) to get all children of a node closer together inside $A_i$. Next, we discuss below how to perform operations in $\CST$ at least as efficiently as in ST.

\subsection{Operations on the Simplex Array List}\label{s63} 
Let us now analyze the cost of performing basic operations on $\CST$ (the motivation behind these operations are well described in \cite{SimplexTree}). Denote by $\Gamma_j(\sigma,\tau)$ the number of maximal simplices that contain a $j$-simplex $\tau$ which is in $\sigma$. Define $\Gamma_j(\sigma)=\underset{\tau}{\max\ }\Gamma_j(\sigma,\tau)$ and $\Gamma_j=\underset{\sigma\in K}{\max\ }\Gamma_j(\sigma)$. It is easy to see that $k\ge \Gamma_0\ge\Gamma_1\ge \cdots \ge \Gamma_d = 1$. In the case of $\CST$, we are interested in the value of $\Gamma_1$ which we use to estimate the worst-case cost of basic operations in $\CST$.  

\noindent\textbf{Membership of Simplex.} To determine membership of $\sigma=v_{\ell_0} \cdot\cdot\cdot v_{\ell_{d_\sigma}}$ in $K$, first determine the contiguous subarrays of $A_{\ell_0},\dots ,A_{\ell_{d_\sigma}}$, say $B_{\ell_0},\dots ,B_{\ell_{d_\sigma}}$ such that every $B_{\ell_i}$ contains all nodes with labels of the form $(\ell_i,\ell_{i+1},z)$, for some $z$ ($B_{\ell_i}$'s indeed form a contiguous subarray because of the way elements in $A_{\ell_i}$ were sorted). We emphasize here that we determine each $B_{\ell_i}$ only by its starting and ending location in $A_{\ell_i}$ and do not explicitly read the contents of each element in $B_{\ell_i}$. Thus, if $P$ is a projection function such that $P((\ell_i,\ell_{i+1},z))=z$ then, we see each $P(B_{\ell_i})$ as a subset of $\{1,\dots ,k\}$ because the only part of the label that distinguishes two elements in $B_{\ell_i}$  is the hash value of the maximal simplex. Now we have $\sigma\in K$ if and only if $\underset{0\le i\le d_\sigma}{\cap}P(B_{\ell_i})\neq\emptyset$. This is because if $\sigma\in K$ then, from Lemma~\ref{LemBij}  there should exist a path corresponding to this simplex which would imply $\underset{i}{\cap}P(B_{\ell_i})\neq\emptyset$, and if $\tau\in \underset{i}{\cap}P(B_{\ell_i})$ then, $\sigma$ is a face of $\tau$. Computing the intersection can be done in $\mathcal{O}(\gamma d_\sigma \log \zeta)$ time, where $\gamma=\underset{i}{\min} \lvert B_{\ell_i}\rvert$ and $\zeta=\underset{i}{\max}\lvert A_{\ell_i}\rvert$. Computing the subarrays can be done in $\mathcal{O}(d_\sigma\log \zeta)$ time. Thus, the total running time is $\mathcal{O}(d_{\sigma}(\gamma\log \zeta+\log\zeta))=\mathcal{O}(d_\sigma \Gamma_1\log(\Gamma_0 d))=\mathcal{O}(d_\sigma \Gamma_1\log(kd))$.

For example, consider the $\CST$ of figure~\ref{fig:CompactSimplexTreeExampleNumberLine} and we have to check the membership of $\sigma=2-3-5$ in the complex of figure~\ref{fig:SimplicialComplexExample}. Then, we have $B_2=\{(2,3,2)\},B_3=\{(3,5,1),(3,5,2)\}$, and $B_5=\{(5,\varphi ,1),(5,\varphi ,2)\}$. We see each $P(B_i)$ as a subset of $\{1,2,3\}$ as follows: $P(B_2)=\{2\},P(B_3)=\{1,2\}$, and $P(B_5)=\{1,2\}$. Clearly $\underset{i}{\cap} P(B_i)=\{2\}\neq \emptyset$, and $\sigma$ is indeed a face of the second maximal simplex $2-3-4-5$.

\noindent\textbf{Insertion.} Suppose we want to insert a maximal simplex $\sigma$ then, building a connected component takes time $\mathcal{O}(d_\sigma^3)$. Updating the arrays $A_i$ takes time $\mathcal{O}(d_\sigma^2\log \zeta)$. Next, we have to check if there exists maximal simplices in $K$ which are now faces of $\sigma$, and remove them. We consider every edge $\sigma_\Delta$ in $\sigma$ and compute $Z_\Delta$ the set of all maximal simplices which contain $\sigma_\Delta$ (which can be done in time $\mathcal{O}( d_{\sigma}^3\Gamma_1 \log (\Gamma_0 d))$). Then, we compute $\underset{\sigma_\Delta\in\sigma}{\cup}Z_\Delta$ whose size is at most $d_\sigma^2\Gamma_1$ and check if any of these maximal simplices are faces in $\sigma$ (can be done in $\mathcal{O}(d_\sigma^3\Gamma_1)$ time). To remove all such faces of $\sigma$ which were previously maximal takes time at most $\mathcal{O}(d_\sigma^4\Gamma_1)$.  Therefore, total time for insertion is $\mathcal{O}(d_\sigma^3\Gamma_1 (d_\sigma+\log (\Gamma_0 d)))=\mathcal{O}(d_\sigma^3\Gamma_1 (d_\sigma+\log (kd)))$.

\noindent\textbf{Removal.} To remove a face $\sigma$, obtain the maximal simplices which contain it (can be done through a membership query in $\mathcal{O}( d_{\sigma}\Gamma_1  \log (\Gamma_0 d))$ time). There are at most $\Gamma_{d_\sigma}$ maximal simplices containing $\sigma$.
Next, remove the above maximal simplices and then insert the facets of the above maximal simplices which do not contain $\sigma$. More precisely, for each of the above obtained maximal simplices make $d_\sigma$ copies of the corresponding connected component, and in the $i^{\textrm{th}}$ copy delete all nodes with label $(\sigma_i,x,y)$ for some $x,y$, and where $\sigma_i$ denotes the label of the $i^{\textrm{th}}$ vertex of $\sigma$. Note that one has to first check if the facets of the maximal simplices not containing $\sigma$ are indeed maximal before inserting them. 
Thus, the total running time is $\mathcal{O}(\Gamma_{d_\sigma}( d_\sigma (d^3\log \Gamma_0 d + d \Gamma_1 \log \Gamma_0 d) )
)=\mathcal{O}(\Gamma_{d_\sigma}dd_\sigma(d^2 +  \Gamma_1 )\log k d)$.

\noindent\textbf{Elementary Collapse.}  A simplex $\tau$ is collapsible through one of its faces $\sigma$,
if $\tau$ is the only coface of $\sigma$. Such
a pair $(\sigma,\tau)$ is called a free pair, and removing both faces of a free pair is an elementary
collapse. Given a pair of simplices $(\sigma,\tau)$, to check if it is a free pair is done by obtaining the list of all maximal simplices which contain $\sigma$, through the membership query (costs $\mathcal{O}(d_\sigma \Gamma_1\log(\Gamma_0 d))$ time) and then checking if $\tau$ is the only member in that list. If yes, remove $\tau$ and insert the facets (except $\sigma$) which are now maximal after removing $\tau$. This takes time $\mathcal{O}(d_\sigma^4+d_\sigma^2\Gamma_1\log (\Gamma_0 d))$. Thus, the total running time is $\mathcal{O}(d_\sigma^2 (d_\sigma^2+\Gamma_1\log(\Gamma_0 d)))=\mathcal{O}(d_\sigma^2 (d_\sigma^2+\Gamma_1\log(kd)))$.

\noindent\textbf{Edge Contraction.} Here we cannot do better than entirely rebuilding the parts of $\CST$ corresponding to maximal simplices which contain the vertices in the edge to be contracted and therefore the cost of the operation is $\mathcal{O}(\Gamma_0 ( d ( \Gamma_1\log (kd) + d^2 ) ))$, as the number of maximal simplices a vertex may be part of is at most $\Gamma_0$ (by definition) and checking if a simplex remains maximal after the edge contraction can be done through the membership query. 

We summarize in Table~\ref{tab:OperationsonSAL} the asymptotic cost of basic operations discussed above and compare it with ST and MxST, through which the efficiency of $\CST$ is established.
\begin{table}[!h]
\begin{center}
\begin{tabular}{|p{4.5cm}|p{2.5cm}|p{2.8cm}|p{4.5cm}|}\hline
\multicolumn{1}{|c|}{}&\multicolumn{1}{|c|}{ST}&\multicolumn{1}{|c|}{$\MxST$}&\multicolumn{1}{|c|}{$\CST$}\\\hline
Storage\tablefootnote{We would like to recapitulate here the lower bound from Theorem~\ref{Lowerbound} of $\Omega(kd\log n)$.}&$\mathcal{O}(k2^d\log n)$\tablefootnote{The space needed to represent ST is $\Theta(m\log n)$ which is written as $\mathcal{O}(k2^d\log n)$ to help in comparison.}&$\mathcal{O}(kd\log n)$&$\mathcal{O}(kd^2(d+\log n+\log k))$\\ \hline
Membership of a simplex $\sigma$&$\mathcal{O}(d_{\sigma}\log n)$ &$\mathcal{O}(kd\log n)$& $\mathcal{O}( d_{\sigma}\Gamma_1  \log (kd))$ \\\hline
Insertion of a simplex $\sigma$ &$\mathcal{O}(2^{d_{\sigma}}d_{\sigma}\log n)$ &$\mathcal{O}(kd\log n)$& $\mathcal{O}(d_\sigma^3\Gamma_1 (d_\sigma+ \log (kd)))$\\\hline
Removal of a face & $\mathcal{O}(m\log n)$ &$\mathcal{O}(kd\log n)$&$\mathcal{O}(\Gamma_{d_\sigma}dd_\sigma(d^2 +  \Gamma_1 )\log k d)$\\\hline
Elementary Collapse & $\mathcal{O}(2^{d_{\sigma}}\log n)$ &$\mathcal{O}(kdd_\sigma\log n)$&$\mathcal{O}(d_\sigma^2 (d_\sigma^2+ \Gamma_1\log(k d)))$\\\hline
Edge Contraction & $\mathcal{O}(md)$ &$\mathcal{O}(kd(k+\log n))$&$\mathcal{O}(\Gamma_0 ( d ( \Gamma_1\log (kd) + d^2 ) ))$ \\\hline
\end{tabular}
\end{center}
\caption{Cost of performing basic operations on SAL in comparison with ST and MxST.}
\label{tab:OperationsonSAL}
\end{table}


\noindent\textbf{Performance of SAL.} Plainly, if the number of maximal simplices is small (i.e., can be considered as a constant), $\CST$ and $\MxST$ are very efficient data structures and this is indeed the case for a large class of complexes encountered in practice as discussed in section~\ref{S2}. 

Remarkably, even if $k$ is not small but $d$ is small then, $\CST$ is a compact data structure as given by the lower bound in Theorem~\ref{Lowerbound}. This is because $\mathcal{O}(kd^2(d+\log n+\log k))$ bits are sufficient to represent $\CST$ and the lower bound is met when $d$ is fixed  (as it translates to needing $\mathcal{O}(k\log n)$ bits to represent $\CST$). Also, it is worth noting here that $\Gamma_0$ is usually a small fraction of $k$ and since $\Gamma_1$ is at most $\Gamma_0$, the above operations are performed considerably faster than in MxST where almost always the only way to perform operations is to traverse the entire tree. Indeed SAL was intended to be efficient in this regard as even if $k$ is not small the construction of SAL replaces the dependence on $k$ by a dependence on a more local parameter $\Gamma_1$  that reflects some ``local complexity'' of the simplicial complex. As a simple demonstration, we estimated $\Gamma_0,\Gamma_1,\Gamma_2,$ and $\Gamma_3$ for the simplicial complexes of Data Set 1 (see section~\ref{Exp}). These values are recorded in Table~\ref{tab:Gamma}.

\begin{table}[!h]
\begin{center}\resizebox{13cm}{!}{
\begin{tabular}{|c|c|c|c|c|c|c|c|c|c|c|c}\hline
No&$n$&$r$&$d$&$k$&$m$&$\Gamma_0$&$\Gamma_1$&$\Gamma_2$&$\Gamma_3$&$\lvert \CST\rvert$
\\\hline
1&10,000&0.15&10&24,970&604,573&62&53&47&37&424,440
\\\hline
2&10,000&0.16&13&25,410&1,387,023&71&61&55&48&623,238
\\\hline
3&10,000&0.17&15&27,086&3,543,583&90&67&61&51&968,766
\\\hline
4&10,000&0.18&17&27,286&10,508,486&115&91&68&54&1,412,310
\\\hline
\end{tabular}}
\end{center}
\caption{Values of $\Gamma_0$, $\Gamma_1$, $\Gamma_2$, and $\Gamma_3$ for the simplicial complexes generated from Data Set 1.}
\label{tab:Gamma}
\end{table}

It is interesting to note that the size of SAL is larger than the size 
of $\CoST$ but much smaller than the size of $\ST$. This is expected, as SAL promises to perform most basic operations more efficiently than ST while compromising slightly on size. Further our intuition, as described previously, was that $\Gamma_0$ should be much smaller than $k$, which is supported by the above results. Also, we note that 
for larger simplicial complexes such as complexes No 3 and 4, there is
a noticeable gap between $\Gamma_0$ and $\Gamma_1$. Since complexity of basic operations using SAL is parametrized by $\Gamma_1$ (and not $\Gamma_0$), the above results support our claim that SAL is an efficient data structure.

\noindent\textbf{Local Sensitivity of Simplex Array List.} It is worth 	noting that while the cost of basic operations are bounded using 
$\Gamma_1$, we could use local parameters such as $\gamma$ and 
$Z_\Delta$ (see previous paragraphs on Membership of Simplex and 
Insertion for definition) to get a better estimate on the cost of 
these operations. $\gamma$ captures local information about a simplex 
$\sigma$ sharing an edge with other maximal simplices of the 
complex. More precisely, it is the minimum, over all edges of 
$\sigma$, of the largest number of maximal simplices that contain the 
edge. If $\sigma$ has an edge which is contained in a few maximal 
simplices then, $\gamma$ is  small. $Z_\Delta$ captures another local 
property of a simplex $\sigma$ -- the set of all maximal simplices 
that contain the edge $\sigma_\Delta$. Therefore, 
SAL is sensitive to the local structure of the complex. 

\subsection{A Sequence of  Representations for Simplicial Complexes}

We can use the operation $\mathcal{R}$ to generate a sequence of data
structures, each more powerful than the previous ones (but also
bulkier). More formally, consider the sequence of data structures
$\langle \Lambda\rangle$, where $\Lambda_{-1}=\MxST$,
$\Lambda_{0}=\mathcal{U}(\MxST)$ and $\Lambda_i =
\mathcal{R}^i(\mathcal{T}(\mathcal{U}(\MxST)))$, for all
$i\in\mathbb{N}$. Note that $\Lambda_1=\CST$. Further, for all $i\in\mathbb{N}$, we will refer to the data structure $\Lambda_i$ by the name $i\mhyphen \CST$ (we will continue to refer to $1\mhyphen\CST$ as $\CST$).

Further, we see that in the $i^{\textrm{th}}$ element of the sequence,
every node which is not a leaf (sink) in the data structure corresponds to a unique $i$-simplex
in the simplicial complex. Also for all $i\mhyphen\CST$, $i\in\mathbb{N}$,
we have that it is a NFA recognizing all the simplices in the
complex. As we move along the sequence, the size of the data structure
blows up by a factor of $d$ at each step. But in return, we gain efficiency in
searching for simplices as the membership query depends on
$\Gamma_i$ which decreases as $i$ increases.

In section~\ref{0-SAL}, we will see how to construct 0-SAL (i.e., $\mathcal{U}(\MxST)$) from the maximal
simplices of a simplicial complex and in section~\ref{2-SAL}, we will see
how to construct $\mathcal{R}(\mathcal{R}(\mathcal{T}(\mathcal{U}(\MxST))))=2\mhyphen\CST$ from
the maximal simplices of a simplicial complex, and this would help to demonstrate the
construction of data structures which appear later in the sequence $\langle \Lambda\rangle$.

\subsubsection{Unprefixed Maximal Simplex Tree}\label{0-SAL}

The \emph{Unprefixed Maximal Simplex Tree} $\mathcal{U}(\MxST)$ or 0-SAL is a directed acyclic graph which can be obtained by modifying MxST or  can be  constructed from the maximal simplices of $K$. We will describe the latter here. We initially have $n$ empty arrays $A_1,\ldots ,A_n$ and for every maximal simplex $\sigma=v_{\ell_0} \cdot\cdot\cdot v_{\ell_j}$, associate a unique key between $1$ and $k$ generated using a hash function $\mathcal{H}$ and insert $\mathcal{H}(\sigma)$ in the arrays $A_{\ell_0},\ldots ,A_{\ell_j}$. Thus determining membership of a simplex again reduces to computing set intersection and insertion and removal of simplices are performed as in MxST but here we are equipped with a  faster search operation (i.e., more efficient membership query). We summarize in Table~\ref{tab:OperationsonUMxST} the asymptotic cost of basic operations for 0-SAL and compare it with $\MxST$. 

\begin{table}[!h]
\begin{center}
\begin{tabular}{|p{4.5cm}|p{4.25cm}|p{4cm}|}\hline
\multicolumn{1}{|c|}{Operation}&\multicolumn{1}{|c|}{Cost for $\MxST$}&\multicolumn{1}{|c|}{Cost for 0-SAL}\\\hline
Membership of a simplex $\sigma$& $\mathcal{O}( kd\log n)$ &$\mathcal{O}( d_{\sigma}\Gamma_0  \log k)$\\\hline
Insertion of a simplex $\sigma$ & $\mathcal{O}(kd\log n)$& $\mathcal{O}(d_\sigma\Gamma_0( d_\sigma+ \log k))$\\\hline
Removal of a face &$\mathcal{O}(kd\log n)$&$\mathcal{O}(d_{\sigma}d\Gamma_0\Gamma_{d_{\sigma}}\log k)$\\\hline
Elementary Collapse &$\mathcal{O}(kdd_\sigma\log n)$&$\mathcal{O}(d_\sigma^2 \Gamma_0\log k)$\\\hline
Edge Contraction &$\mathcal{O}(kd(k+\log n))$ &$\mathcal{O}(\Gamma_0^2d\log k)$\\\hline
\end{tabular}
\end{center}
\caption{Cost of performing basic operations on 0-SAL which is compared with cost of such operations on MxST.}
\label{tab:OperationsonUMxST}
\end{table}

1-$\CST$ has two advantages over 0-SAL. First, 1-$\CST$
stores all simplices through its paths (Lemma~\ref{LemBij}) and this
may help in storing filtrations. In addition, it is likely that
$\Gamma_1$ is significantly smaller than $\Gamma_0$. 

\subsubsection{2$\mhyphen$Simplex Array List}\label{2-SAL}
The {\em 2$\mhyphen$Simplex Array List} $\TSC(K)$ is a directed acyclic graph which can be obtained by modifying MxST (i.e., $\TSC=\mathcal{R}(\mathcal{R}(\mathcal{T}(\mathcal{U}(\MxST))))$) or  can be  constructed from the maximal simplices of $K$. We will describe the latter here. For a given maximal simplex $\sigma=v_{\ell_0} \cdot\cdot\cdot v_{\ell_j}$, associate a unique key between $1$ and $k$ generated using a hash function $\mathcal{H}$ and then introduce $\frac{j(j^2+5)}{6}+1$ new nodes in $\TSC$. We build a set of $\frac{j(j^2+5)}{6}+1$ labels and assign uniquely a label to each node. The set of labels is defined as the union of the following three sets
 (cf.\ Figure~\ref{fig:3CompactSimplexTreeExampleNumberLine} for an example):\\
$S_1=\{(\ell_{i_1},(\ell_{i_2},\ell_{i_3}),\mathcal{H}(\sigma))\mid i_1\in\{0,1,\dots ,j-2\},i_2\in\{i_1+1,\dots ,j-1\},i_3\in\{i_2+1,\dots ,j\}\}$\\
$S_2=\{(\ell_i,(\ell_j,\varphi),\mathcal{H}(\sigma))\mid i\in\{0,1,\dots ,j-1\}\}$\\
$S_3=\{(\ell_j,(\varphi,\varphi),\mathcal{H}(\sigma))\}$
\begin{sloppypar}\noindent where $\varphi$ denotes an empty label. Now introduce an edge between two nodes with labels $(\ell_{i_1},(\ell_{i_2},\ell_{i_3}),\mathcal{H}(\sigma))$ and $(\ell_{i_4},(\ell_{i_5},\ell_{i_6}),\mathcal{H}(\sigma))$ if and only if $i_2=i_4$ and $i_3=i_5$. Next, introduce an edge between two nodes with labels $(\ell_{i_1},(\ell_{i_2},\ell_{i_3}),\mathcal{H}(\sigma))$ and $(\ell_{i_4} ,(\ell_{j},\varphi),\mathcal{H}(\sigma))$ if and only if $i_2=i_4$ and $i_3=j$. Finally, introduce an edge from every node with label in $S_2$ to the node with label in $S_3$. Thus, in $\TSC$ we represent a maximal $j$-simplex using a connected component containing $|S_1|+|S_2|+|S_3|=\frac{j(j^2+5)}{6}+1$ nodes. To perform basic operations efficiently, embed $\TSC$ on the number line such that for every $i\in\{1,2,\dots ,n\}$ on the number line we have an array $A_i$ of nodes which has labels of the form $(i,(a,b),z)$ for some $z\in\{1,\dots ,k\}$ and $a,b\in\{i+1,\dots ,n,\varphi\}$. Sort each $A_i$ based on $a$ and in case of ties, sort them based on $b$ and in case of further ties sort them based on $z$.\end{sloppypar}

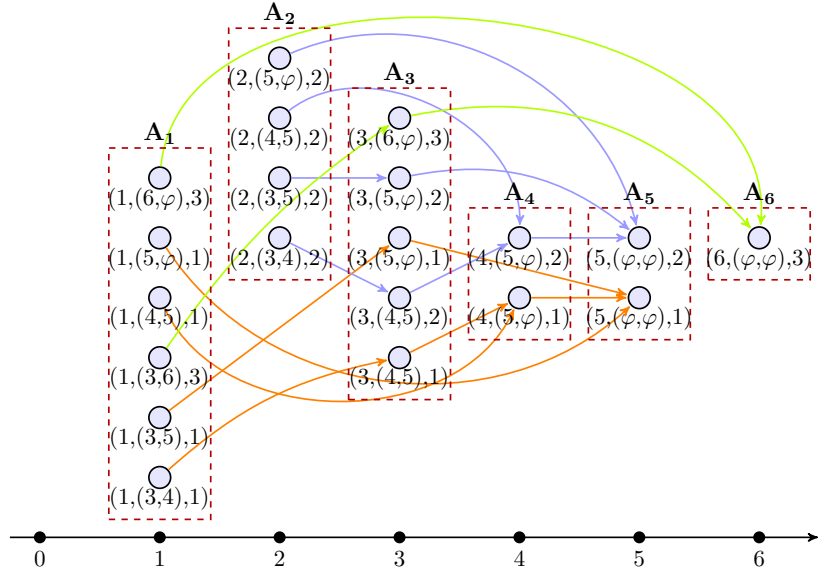
\begin{figure}[!h]
\centering
\resizebox{11cm}{!}{
 
\begin{tikzpicture}[->,>=stealth',shorten >=0.5pt,auto,node distance=2cm,
  thick,main node/.style={circle,fill=blue!10,draw,font=\sffamily\large\bfseries}]
  \node[main node] (1134) at (1,1) {};
  \node[main node] (1135) at (1,2) {};
  \node[main node] (1336) at (1,3) {};  
  \node[main node] (1145) at (1,4) {};
  \node[main node] (115p) at (1,5) {};    
  \node[main node] (136p) at (1,6) {};    
  \node[main node] (2234) at (3,5) {};
  \node[main node] (2235) at (3,6) {};  
  \node[main node] (2245) at (3,7) {};
  \node[main node] (225p) at (3,8) {};  
  \node[main node] (3145) at (5,3) {};
  \node[main node] (3245) at (5,4) {};  
  \node[main node] (315p) at (5,5) {};  
  \node[main node] (325p) at (5,6) {};
  \node[main node] (336p) at (5,7) {};    
  \node[main node] (415p) at (7,4) {};
  \node[main node] (425p) at (7,5) {};
  \node[main node] (51pp) at (9,4) {};
  \node[main node] (52pp) at (9,5) {};
  \node[main node] (63pp) at (11,5) {};
      
  \path[every node/.style={font=\sffamily}]
    (1134) edge [red!50!yellow, bend left=15] (3145)
    (1135) edge [red!50!yellow] (315p)
    (1145) edge [red!50!yellow, bend right=65] (415p)    
    (115p) edge [red!50!yellow, bend right =50] (51pp)
    (3145) edge [red!50!yellow] (415p)
    (315p) edge [red!50!yellow] (51pp)
    (415p) edge [red!50!yellow] (51pp)
    (2234) edge [blue!40!white] (3245)
    (2235) edge [blue!40!white] (325p)
    (2245) edge [blue!40!white, bend left=65] (425p)    
    (225p) edge [blue!40!white, bend left =50] (52pp)
    (3245) edge [blue!40!white] (425p)
    (325p) edge [blue!40!white,bend left =25] (52pp)
    (425p) edge [blue!40!white] (52pp)
    (1336) edge [green!30!yellow, bend left=10] (336p)    
    (136p) edge [green!30!yellow,bend left =90] (63pp)
    (336p) edge [green!30!yellow, bend left =30] (63pp)
                    
    
     ;
     \draw [->] (-1.5,0) -- (12,0);
     \fill[black] (-1,0) circle (0.1);
     \fill[black] (1,0) circle (0.1);     
     \fill[black] (3,0) circle (0.1);
     \fill[black] (5,0) circle (0.1);     
     \fill[black] (7,0) circle (0.1);
     \fill[black] (9,0) circle (0.1);     
     \fill[black] (11,0) circle (0.1);
     \node at (-1,-0.35) {0};
     \node at (1,-0.35) {1};
     \node at (3,-0.35) {2};
     \node at (5,-0.35) {3};
     \node at (7,-0.35) {4};
     \node at (9,-0.35) {5};
     \node at (11,-0.35) {6};
     \draw[red!70!black,dashed] (3.85,8.5) rectangle (2.15,4.3);
     \draw[red!70!black,dashed] (1.85,6.5) rectangle (0.15,0.3);     
     \draw[red!70!black,dashed] (4.15,2.3) rectangle (5.85,7.5);
     \draw[red!70!black,dashed] (6.15,3.3) rectangle (7.85,5.5);     
     \draw[red!70!black,dashed] (8.15,3.3) rectangle (9.85,5.5);
     \draw[red!70!black,dashed] (10.15,4.3) rectangle (11.85,5.5);     
	 \node at (1,0.65) {(1,(3,4),1)};
     \node at (1,1.65) {(1,(3,5),1)};
     \node at (1,3.65) {(1,(4,5),1)};          
     \node at (1,4.65) {(1,(5,$\varphi$),1)};
     \node at (1,2.65) {(1,(3,6),3)};
     \node at (1,5.65) {(1,(6,$\varphi$),3)};          
	 \node at (3,4.65) {(2,(3,4),2)};
     \node at (3,5.65) {(2,(3,5),2)};
     \node at (3,6.65) {(2,(4,5),2)};          
     \node at (3,7.65) {(2,(5,$\varphi$),2)};
     \node at (5,2.65) {(3,(4,5),1)};          
	 \node at (5,4.65) {(3,(5,$\varphi$),1)};
     \node at (5,3.65) {(3,(4,5),2)};
     \node at (5,5.65) {(3,(5,$\varphi$),2)};          
     \node at (5,6.65) {(3,(6,$\varphi$),3)};
     \node at (7,3.65) {(4,(5,$\varphi$),1)};
     \node at (7,4.65) {(4,(5,$\varphi$),2)};          
     \node at (9,3.65) {(5,($\varphi$,$\varphi$),1)};
     \node at (9,4.65) {(5,($\varphi$,$\varphi$),2)};          
     \node at (11,4.65) {(6,($\varphi$,$\varphi$),3)};
	\large{\node at (1, 6.75) {$\mathbf{A_1}$};     }
	\large{\node at (3, 8.75) {$\mathbf{A_2}$};     }
	\large{\node at (5, 7.75) {$\mathbf{A_3}$};     }
	\large{\node at (7, 5.75) {$\mathbf{A_4}$};     }
	\large{\node at (9, 5.75) {$\mathbf{A_5}$};     }
	\large{\node at (11, 5.75) {$\mathbf{A_6}$};     }		
				
\end{tikzpicture}
	}
\caption{2 - Simplex Array List for complex in Figure~\ref{fig:SimplicialComplexExample} embedded on the number line.}

  \label{fig:3CompactSimplexTreeExampleNumberLine}

\end{figure}

$\TSC(K)$ has at most $k\left(\frac{d(d^2+5)}{6}+1\right)$ nodes and,  for each maximal simplex of dimension $d_\sigma$, the outdegree of any node in the connected component is at most $d_\sigma$. Therefore, the total number of edges in $\TSC(K)$ is at most $k\left(\frac{d^2(d^2+5)}{6}+d\right)$. Hence, the space required to store $\TSC(K)$ is $\mathcal{O}(kd^3(d+\log n+\log k))$. 

We summarize in Table~\ref{tab:OperationsonTSC} the asymptotic cost of basic operations for 2 - SAL and compare it with $\CST$. All the basic operations are performed similar to the way we did in $\CST$ except that here we have more structure and thus searching becomes more efficient ($\Gamma_1$ is replaced by the smaller $\Gamma_2$) but pay extra in size because we have to maintain the additional structure.
\begin{table}[!h]
\begin{center}
\begin{tabular}{|p{5.5cm}|p{4.75cm}|p{4,75cm}|}\hline
\multicolumn{1}{|c|}{Operation}&\multicolumn{1}{|c|}{Cost for $\CST$}&\multicolumn{1}{|c|}{Cost for $\TSC$}\\\hline
Membership of a simplex $\sigma$&$\mathcal{O}( d_{\sigma}(\Gamma_1 \log (kd)))$  & $\mathcal{O}( d_{\sigma}\Gamma_2  \log (kd^2))$ \\\hline
Insertion of a maximal simplex $\sigma$ &$\mathcal{O}(d_\sigma^3(\Gamma_1 d_\sigma+ \log (kd)))$& $\mathcal{O}(d_\sigma^3\Gamma_2( d_\sigma^3+ \log (kd^2)))$\\\hline
Removal of a face &$\mathcal{O}(d_{\sigma}d^3(\Gamma_1+\log(kd)))$&$\mathcal{O}(\Gamma_{d_\sigma}dd_\sigma(d^3 +  \Gamma_2 )\log (k d^2))$\\\hline
Elementary Collapse & $\mathcal{O}(d_\sigma^2 (d_\sigma^2+ \Gamma_1+\log(kd^2)))$ &$\mathcal{O}(d_\sigma^2 (d_\sigma^3+ \Gamma_2\log(kd^2)))$\\\hline
Edge Contraction &$\mathcal{O}(\Gamma_0 ( d ( \Gamma_1\log (kd) + d^2 ) ))$  &$\mathcal{O}(\Gamma_0 ( d ( \Gamma_2\log (kd^2) + d^3 ) ))$ \\\hline
\end{tabular}
\end{center}
\caption{Cost of performing basic operations on $\TSC$ which is compared with cost of such operations on SAL.}
\label{tab:OperationsonTSC}
\end{table}

\subsubsection{The Bottom Line}
A natural question to resolve is which element in
$\langle\Lambda\rangle$ should one pick for representing a simplicial
complex. This indeed depends on the nature of data and the type of
complex. For instance, consider the case of a Rips complex whose
vertex set is a  good  sample of a smooth manifold of low intrinsic dimension. Then, we would settle for either $\Lambda_0$ or $\Lambda_1$ as we expect $\Gamma_0$ itself to be quite low. However if the simplicial complex is of high dimension and if there are central simplices on which most maximal simplices are built on then, it might be better to look at elements higher up in the sequence as $\Gamma$ might be quite high in the beginning and can collapse quickly after some point.

\section{Labeling Dependency}\label{S8}
In this section, we discuss how the labeling of the vertices affects the size of the data structures discussed in this paper. In particular, the size of both ST and SAL are invariant over the labeling of vertices in a simplicial complex (see Lemma~\ref{CSTinvariant}). However, this is not the case with MxST and $\mathcal{M}(\SA)$. To see this, consider a simplicial complex which contains a  maximal triangle and a maximal tetrahedron, sharing an edge. We could label the triangle and tetrahedron as 1--2--3 and 1--2--4--5, or as 1--3--4 and 2--3--4--5 respectively. Note that the two labelings give two $\mathcal{M}(\mathrm{SA})$ (and MxST) of different sizes.

In all of the hardness results we will see, the reduction will be from
vertex cover for some specific graphs, where graphs are considered as
1-dimensional simplicial complexes. More formally, given a graph $G$ on vertex set $V$ and edge set $E$, the associated 1-dimensional simplicial complex is defined on the vertex set $V$ which admits all edges in $E$ as 1-dimensional simplices. When there is no confusion, we will refer to this simplicial complex also as $G$.

We observe that MxST$(G)$ is of height 2. In layer 0 we have the root, and in layer 2 we have all the leaves (we may have some leaves in layer 1 if there are vertices in $G$ of degree zero). Since the number of leaves in the MxST is equal to the number of maximal simplices in $G$ which is in turn equal to the number of edges in $G$ (plus the number of zero-degree vertices), regardless of the ordering on vertices in $G$, the number of nodes in layer 0 and 2 are fixed to 1 and $\lvert E\rvert$ respectively. Further note that vertices on layer 1 form a vertex cover of $G$. 

\sloppypar{We will similarly analyze $\mathcal{C}(\ST(G))$. $\ST(G)$ can be obtained from $\MxST(G)$ by introducing nodes on layer 1 (and edges from the root to these nodes) such that all vertices appear in nodes of layer 1 of $\ST(G)$. This means that in $\ST(G)$, we have the root in layer 0, $\lvert V\rvert$ nodes on layer 1, and $\lvert E\rvert$ nodes on layer 2. We recapitulate here that the size of $\ST$ is invariant over the ordering of vertices (as can be observed above -- $\lvert\ST(G)\rvert=\lvert V\rvert+\lvert E\rvert$). More importantly, for 1 - dimensional simplicial complexes we have that $\lvert \CoST\rvert=\lvert \ST\rvert$. Therefore to analyze the impact of labeling on the size of $\CoST$, we will have to move to higher dimensional complexes (which we indeed do).}

\subsection{Optimal Labeling for the Maximal Simplex Tree}
The size of MxST is very sensitive to the labeling of the vertices. For instance, in Example~\ref{Algo5Optimal}, by reversing the labeling of vertices, we increase the size of MxST($K$) by a factor of order $k$. 
\begin{example}\label{Algo5Optimal}
Let $K\in {\cal K} (d+k,k,d,m)$ whose set of maximal simplices is $\{1,2,\dots ,d,d+i|1\le i\le k\}$.
\end{example}

First, we formalize the label ordering problem on MxST: Given an integer $\alpha$ and a simplicial complex $K_{\theta}\in {\cal K}_\theta(n,k,d,m)$, does there exist a permutation $\pi$ of $1,2,\dots ,n$ such that $\lvert \MxST(K_{\pi \circ \theta})\rvert \le \alpha$? Let us refer to this problem as \MxSTMin$(K_\theta,\alpha)$ and from Theorem 17 of \cite{BM14}, we know that it is NP-Complete even for 1-dimensional complexes. 
\begin{theorem}[\cite{BM14}]\label{MxSTLabel}
\MxSTMin\ is NP-Complete.
\end{theorem}
\begin{proof}
Clearly \MxSTMin\ is in NP. We will now show that it is NP-Hard as well. Given a connected graph $G$ on vertex set $V$ (with $\theta$ being a labeling of the vertices from $1$ to $\lvert V\rvert$) and edge set $E$, we define $K_{\theta}$ to be the 1-dimensional simplicial complex associated to $G$. Since $G$ is connected, we have that all the leaves in $\MxST(K_\theta)$ appear in layer 2. Thus finding a vertex cover for $G$ of size at most $\alpha$ is equivalent to finding an ordering $\pi$ on the vertices in $K_{\theta}$ (to determine which of these should appear in nodes in layer 1 of the MxST) such that $\lvert\MxST(K_{\pi \circ \theta})\rvert\le \alpha +\lvert E\rvert$. Since vertex cover problem is NP-Hard \cite{K72}, we have that \MxSTMin\ is also NP-Hard. Thus \MxSTMin\ is NP-Complete.
\end{proof}

Intuitively, finding a good labeling for $\mathcal{C}(\MxST)$ seems
harder than for MxST. More formally, given an integer $\alpha$ and a
simplicial complex $K_{\theta}\in {\cal K}_\theta(n,k,d,m)$, does
there exist a permutation $\pi$ of $1,2,\dots ,n$ such that $\lvert
\mathcal{C}(\MxST(K_{\pi \circ \theta}))\rvert \le \alpha$? Let us
refer to this problem as
\CMxSTMin$(K_\theta,\alpha)$. Corollary~\ref{CMxSTLabel} easily
follows  from Theorem~\ref{MxSTLabel} as, for any 
1-dimensional simplicial complex and any fixed labeling, 
$\lvert\mathcal{C}(\MxST)\rvert = \lvert \MxST\rvert$.

\begin{corollary}\label{CMxSTLabel}
\CMxSTMin\ is NP-Complete.
\end{corollary}
\begin{proof}
It is clear that \CMxSTMin\ is in NP.
We know from the proof of Theorem~\ref{MxSTLabel} that it is NP-Hard
to decide \MxSTMin\ even for pure simplicial complexes of dimension
1. We will thus show a reduction from \MxSTMin\ for such simplicial
complexes to \CMxSTMin. For any pure simplicial complex $K$ of
dimension 1 under any labelling $\theta$ of its vertices, the size of MxST and $\CoMST$ for $K_\theta$ are the same. Thus any solution for an instance of \MxSTMin\ is also a solution for the same instance for \CMxSTMin\ and vice versa.
\end{proof}

\subsection{Optimal Labeling for the  Minimal Simplex Automaton}

To prove that finding a good
labeling for $\MSA$ is hard, we use  a  reduction from
another instance of
vertex cover for a special class of graphs. We say a graph $G$ is square-free if, for every two vertices $u,v$ in $G$, the number of common neighbors in $G$ is at most 1.
\begin{lemma}\label{PoorHard}
Vertex Cover problem on square-free graphs is NP-Hard.
\end{lemma}
\begin{proof}
We observe that in the reduction from SAT to 3-SAT (Theorem 3.1. of \cite{vertexCoverCubic}), every clause has exactly three distinct variables\footnote{$x$ and $\neg x$ are considered as the same variable.}. Next, we observe that if every clause has exactly three distinct variables then, in the reduction  from 3-SAT to Vertex Cover (Theorem 3.3. of \cite{vertexCoverCubic}), the graph constructed does not have a cycle of length four.
\end{proof}

We will now formalize the decision problem for $\MSA$. Given an integer $\alpha$ and $K_{\theta}\in {\cal K}_\theta(n,k,d,m)$, does there exist a permutation $\pi$ of $1,2,\dots ,n$ such that $\lvert \MSA(K_{\pi \circ \theta})\rvert \le \alpha$? Let us refer to this problem as \MSAMin$(K_\theta,\alpha)$.  We work with  square-free graphs because by using such graphs for building simplicial complexes, we would overcome scenarios in which we have two states in $\MSA$ with identical set of outgoing transitions but have different sets of incoming transitions. 
\begin{theorem}\label{MSALabel}
\MSAMin\ is NP-Complete.
\end{theorem}
\begin{proof}
It is clear that \MSAMin\ is in NP. We will now show that it is
NP-Hard as well. Given a square-free graph $G$ on vertex set $V$
and edge set $E$,  and a labeling $\theta$ from $1$ to $\lvert
V\rvert$  on the vertices, we define  $K_{\theta}$ to be the 1-dimensional simplicial complexes associated to $G$. In the following paragraphs we will prove that $G$ has a vertex cover of size at most $\alpha$  if and only if there exists an ordering $\pi$ on the vertices of $K_{\theta}$ such that $\lvert\mathcal{M}(\SA(K_{\pi \circ \theta}))\rvert\le \alpha +2$.

First, we prove the forward direction. We define the height of a state
as the length of the longest sequence of transitions from the initial
state to that state. We define the height of an automaton as the
maximum over the height of all states in the automaton. Thus
$\mathcal{M}(\SA(K_{\theta}))$ is of height 2. At height 0 we have the
initial state, and at height 2 we have a single state. Transitions
from states in height 1 to height 2 correspond to the edges of $G$
regardless of the ordering on the vertices in $K_{\theta}$. Note that vertices at height 1 form a vertex cover of $G$. Thus if $G$ has a vertex cover of size at most $\alpha$ then, we can construct an ordering $\pi$ on the vertices of $K_{\theta}$ (by allowing all vertices appearing in the vertex cover to appear before the remaining vertices) such that $\lvert\mathcal{M}(\SA(K_{\pi \circ \theta}))\rvert\le \alpha +2$.

Now, we prove the reverse direction. Suppose there exists an ordering
$\pi$ on the vertices of $K_\theta$ such that
$\lvert\mathcal{M}(\SA(K_{\pi \circ \theta}))\rvert\le \alpha +2$. If
there is a state $s$ in $\mathcal{M}(\SA(K_{\pi \circ \theta}))$ at
height 1 which has more than one incoming transition then, it cannot
have more than one outgoing transition because $G$ is square-free. 
In this case, the outgoing transition of $s$ is rewired to be the incoming transition from the initial state and the incoming transitions are rewired to be the outgoing transitions from $s$ to the single
state at height 2. Once this swapping (i.e., rewiring) is done for each state at height 1, we have ensured that the number of incoming transitions for such states is one (because the number of outgoing transitions from these states before rewiring was one).
Additionally, we note that by performing the above rewiring in $\mathcal{M}(\SA(K_{\pi \circ \theta}))$ we have not introduced (or removed) any new states, and thus the size of $\mathcal{M}(\SA(K_{\pi \circ \theta}))$ has not changed (and it is computing the same language as before). Therefore, choosing the set of all labels of outgoing
transitions from the initial state to states at height 1 gives a
subset of the vertex set of size at most $\alpha$ and does indeed form
a vertex cover of $G$. From Lemma~\ref{PoorHard} we know that vertex
cover is NP-Hard for square-free graphs, which implies that \MSAMin\ is also NP-Hard. Thus \MSAMin\ is NP-Complete.
\end{proof}

We believe that to obtain reasonably small sized MxST and $\MSA$ one has to label vertices in decreasing order of $k_v\overset{\textrm{def}}{=}$ number of maximal simplices containing vertex $v$. Thus, in practice one may use this heuristic to find a good labeling.

\subsection{Optimal Labeling for the Compressed Simplex Tree}
We had observed in Section~\ref{S4} the delicate relationship between the sizes of $\MSA$ and $\CoST$. Additionally, we have that the complexity measure of sizes for $\MSA$ and $\CoST$ are not coherent and thus, we will not be able to use Theorem~\ref{MSALabel} here to prove hardness result.
On the other hand, we cannot (trivially) extend the reduction from \MxSTMin\ to
that of finding good labelings for $\CoST$ because, as stated earlier, for 1-dimensional simplicial complexes we have $\lvert \CoST\rvert=\lvert \ST\rvert$. Instead we append the simplicial complex (associated to) $G$ and construct 2-dimensional complexes which we then use to
prove hardness result for finding good labelings for $\CoST$.

 Let us first formalize the decision problem. Given an integer $\alpha$ and a
simplicial complex $K_{\theta}\in {\cal K}_\theta(n,k,d,m)$, does
there exist a permutation $\pi$ of $1,2,\dots ,n$ such that $\lvert
\mathcal{C}(\ST(K_{\pi \circ \theta}))\rvert \le \alpha$? Let us refer
to this problem as \CSTMin$(K_\theta,\alpha)$. To prove the hardness
result for \CSTMin, we provide a reduction from a special instance of
vertex cover problem which will be shown to be NP-hard. Specifically, we
restrict vertex cover problem to  a special class of graphs we call as the King-Maker
graphs. A graph $G$ is called a King-Maker graph if there exists two vertices $u,v$ in $G$ such that $u$ is connected to all vertices in $G$ (and thus we shall fondly refer to this vertex as the king) and $v$ is of degree 1.
\begin{lemma}\label{KMHard}
Vertex Cover problem on King-Maker graphs is NP-Hard.
\end{lemma}
\begin{proof}
Given a graph $G$ on vertex set $V$ and edge set $E$, we build a
King-Maker graph $G^\prime$ by adding two new vertices $u$ and $v$, and
adding an edge between every vertex in $G$ and $u$, and an edge between
$u$ and $v$. If $G$ has a vertex cover of size $\alpha$ then,
$G^\prime$ has a vertex cover of size $\alpha+1$ and vice
versa.
\end{proof}

We are now equipped to prove NP-Hardness of \CSTMin\ and follow a reduction similar to that described in proof of Theorem~\ref{MxSTLabel} and the reduction from King-Maker graphs, helps us ensure that (i) there exists a vertex cover of smallest size which contains the king, and (ii) every vertex of the graph (except the king) appears in nodes of layer 2 of $\CoST$. Also, we append the simplicial complex (associated to) $G$ by introducing a new vertex and extend $G$ through insertion of triangles. 

\begin{theorem}\label{CSTLabel}
\CSTMin\ is NP-Complete.
\end{theorem}
\begin{proof}
  It is clear that \CSTMin\ is in NP. We will now show that it is
  NP-Hard as well. Let $G$ be a King-Maker graph on vertex set $V$ and
  edge set $E$, and $\theta$ be a labeling of the vertices from
  $1$ to $\lvert V\rvert$. We define $K_{\theta}$ to be the 1-dimensional simplicial complex associated to $G$. We then introduce a new vertex $\v$ in $K_{\theta}$ with label $\lvert V\rvert+1$ and form
  maximal triangles, each consisting of $\v$ together with a maximal
  edge of $K_\theta$. We denote this appended $K_\theta$ by $K_\theta^+$. Thus, we still have the same number of maximal
  simplices but the dimension has increased by 1 and
  $\mathcal{C}(\ST(K_{\theta}^+)$ is of height 3 (see Figure~\ref{fig:construction} for a demonstration). The rest of the proof will focus on proving the following claim: $G$ has a vertex cover of size at most $\alpha$ if and only if there exists a permutation $\pi^+$ of the vertices of $K_\theta^+$ such that $\lvert\mathcal{C}(\ST(K_{\pi^+\circ\theta}^+)\rvert$ is at most $2\lvert V\rvert +\lvert E\rvert +\alpha$.
  
  \begin{figure}[!h]
\centering
\resizebox{\linewidth}{!}{
\begin{tikzpicture}[-,>=stealth',shorten >=0.5pt,auto,node distance=2cm,
  thick,main node/.style={circle,draw,font=\sffamily\large\bfseries,minimum size=1pt}]

\node at (-5.8,3.8) {\huge $G$};
\node at (1.2,3.8) {\huge $K^\prime$};
\node at (8,3.8) {\huge $\mathcal{C}(\ST(K^\prime))$};

 \node[main node] (1) at (-5.8,1.8) {1};
  \node[main node] (2) at (-7,0) {2};
   \node[main node] (3) at (-8,2) {3};
    \node[main node] (4) at (-4.1,2.6) {4};
  \path[every node/.style={font=\sffamily\small}]
    (1) edge node {} (2)
    (1) edge node {} (3)
    (1) edge node {} (4)    
    (2) edge node {} (3);

\draw [fill=blue!12] (0,0)--(2.7,-0.8)--(-1,2)--cycle;
\draw [fill=blue!12] (-1,2)--(2.7,-0.8)--(1.2,1.8)--cycle;
\draw [fill=blue!10] (0,0)--(2.7,-0.8)--(1.2,1.8)--cycle;
\draw [fill=blue!8] (2.9,2.6)--(2.7,-0.8)--(1.2,1.8)--cycle;

\fill[black] (0,0) circle (0.1);
\fill[black] (-1,2) circle (0.1);
\fill[black] (1.2,1.8) circle (0.1);
\fill[black] (2.9,2.6) circle (0.1);
\fill[black] (2.7,-0.8) circle (0.1);

\draw [ultra thick] (0,0) -- (-1,2) -- (1.2,1.8) -- (2.9,2.6);
\draw [ultra thick] (0,0) -- (1.2,1.8);
\draw [ultra thick] (0,0) -- (2.7,-0.8) -- (1.2,1.8);
\draw [ultra thick] (2.7,-0.8) -- (2.9,2.6);
\draw [dashed, ultra thick] (2.7,-0.8) -- (-1,2);
\draw [ultra thick] (0.55,0.83) -- (-1,2);

\node at (-1.1,2.35) {3};
\node at (1.25,2.15) {1};
\node at (0,-0.35) {2};
\node at (3,2.95) {4};
\node at (2.8,-1.15) {5};

\draw [blue!40!white, dashed] (5.7,1.95) rectangle (11,1.05);

  \node[main node] (0) at (8,2.5) {X};
  \node[main node] (1) at (7,1.5) {1};
  \node[main node] (2) at (10,1.5) {2};
  \node[main node] (12) at (6,0.5) {2};
  \node[main node] (13) at (7,0) {3};
  \node[main node] (14) at (8,0.5) {4};
  \node[main node] (125) at (6,-1) {5};
      
  \path[->,every node/.style={font=\sffamily\small}]
    (0) edge node {} (1)
    (0) edge node {} (2)
    (1) edge node {} (12)
    (1) edge node {} (13)
    (1) edge node {} (14)        
    (2) edge [red, bend left = 30] node {} (13)
    (12) edge node {} (125)
    (0) edge [red, bend right = 30] node {} (12)
    (0) edge [red] node {} (13)
    (0) edge [red] node {} (14)
    (0) edge [red, bend right = 60] node {} (125)                
    (1) edge [red] node {} (125)                
    (2) edge [red, bend left = 30] node {} (125)                        
    (13) edge [red] node {} (125)                
    (14) edge [red, bend left = 20] node {} (125)                

     ;
\node at (12.5,1.75) {Vertex Cover};
\node at (12.5,1.25) {of $G$};




     
\end{tikzpicture}
}
\caption{Demonstration of the construction of $K^\prime$ through an example.}
\label{fig:construction}
\end{figure}
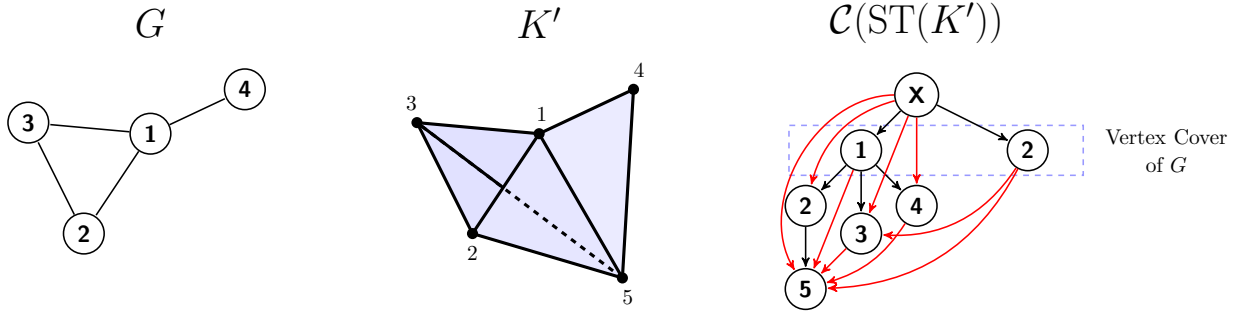
  
  Consider a vertex cover $V^\prime$ of
  $G$ of size at most $\alpha$. We construct a
  permutation $\pi$ of the labels of the vertices of $K_\theta$ so that all
  vertices in $V^\prime$ appear before any other vertex in
  $V$. Further, we assume that the king is in $V^\prime$. This is
  no loss of generality since, if a vertex cover excludes the king, we
  can always find another vertex cover of the same size which includes
  the king. Further, we set
  the label of the king under $\pi$ to $1$. In
  $\mathcal{C}(\ST(K_{\pi\circ\theta}))$, we distinguish three
  layers. Layer 1 consists of the root, the vertices
  of $V'$ appear in layer 2, and all vertices but the king appear in
  layer 3. There are $\lvert V^\prime\rvert$ edges from layer 1 to layer 2, $\lvert E\rvert$ edges from layer 2 to layer 3
  and $\lvert V\rvert - \lvert V^\prime\rvert$ edges from layer 1 to
  layer 3. Now, we extend $\pi$ to construct a permutation $\pi^+$ of
  labels of the vertices of $K_\theta^+$ by setting $\pi(\lvert
  V\rvert +1)=\lvert V\rvert +1$ and keeping the same labels as in
  $\pi$ for the remaining vertices.  We can see $\mathcal{C}(\ST(K_{\pi^+\circ\theta}^+))$
as  adding  one new node and some edges to
$\mathcal{C}(\ST(K_{\pi\circ\theta}))$.  Node with label $\lvert
V\rvert +1$ is put on layer 4 and we introduce an edge from all
vertices in layers 1, 2 and 3  to this node in layer 4. 
Thus we have
$\lvert\mathcal{C}(\ST(K_{\pi^+\circ\theta}^+))\rvert = \lvert\mathcal{C}(\ST(K_{\pi\circ\theta}))\rvert + 1 + \lvert V^\prime\rvert + (\lvert V\rvert -1) = 2\lvert V\rvert +\lvert E \rvert +\lvert V^\prime\rvert\le 2\lvert
V\rvert + \lvert E\rvert + \alpha $. \\


To prove the other direction of the equivalence, we claim that suppose the vertices of $K_\theta^+$ can be labeled under a permutation $\pi^+$ such that the size of $\CoST$ is at most $2\lvert
V\rvert + \lvert E\rvert + \alpha$ then, we can find a vertex cover of
$G$ of size at most $\alpha$. In order to find this vertex cover, we
observe that the structure of an optimal $\CoST$ under $\pi^+$ will
have only one node in layer 4 and this will contain the vertex
$\v$, i.e., 
the new  vertex introduced in $K_\theta$ with label $|V|+1$.
This is
formally stated below as Claim~\ref{OrderingStructure}. Thus, the
structure of $\CoST$ is the one we had previously encountered and can
extract the vertex cover of $G$ by looking at the nodes on layer 2
(layer 2 can have at most $\alpha$ nodes since we will observe later
that in an optimal $\CoST$, the king will appear with label 1 under
$\pi^+$). We will now see that under any permutation if the label of $\v$ is not in the last position then, we can find a new permutation which leads to $\CoST$ of smaller size with the label of $\v$ being at the last position.

\begin{claim}\label{OrderingStructure}
Given a permutation $\pi^+$ such that $\pi^+(\lvert V\rvert +1)\neq \lvert V\rvert +1$, we can construct a new permutation $\rho^+$ which is obtained by moving $\lvert  V\rvert+1$ to the last position such that\footnote{Here we see $\pi^+$ and $\rho^+$ as strings of length $\lvert  V\rvert+1$ obtained by their action on $1,2,\dots ,\lvert V\rvert +1$, and by 'moving' we mean to remove an element from its current position and insert it in the new position.}: $$\lvert\mathcal{C}(\ST(K_{\pi^+\circ\theta}^+))\rvert\ge \lvert\mathcal{C}(\ST(K_{\rho^+\circ\theta}^+))\rvert.$$ 
\end{claim}

Thus, if we have a permutation $\pi^+$ such that
$\lvert\mathcal{C}(\ST(K_{\pi^+\circ\theta}^+))\rvert<2\lvert
V\rvert + \lvert E\rvert + \alpha $, we may assume that $\pi^+(\lvert
V\rvert+1)=\lvert V\rvert+1$ and we can select the vertex cover of $G$
by considering the vertices on layer 2 of
$\mathcal{C}(\ST(K_{\pi^+\circ\theta}^+))$. Next, we argue that in
$\pi^+$, we may assume $\pi^+(1)=1$ (i.e., the king gets label 1 under $\pi^+$), since otherwise, we can
always move it to the first position and it does not affect the size
of $\CoST$. It is now easy to see that the size of layer 2 is less than $\alpha$ as the king is on layer 2.

We know from Lemma~\ref{KMHard} that the vertex cover problem for King-Maker graphs is NP-Hard, and this implies that \CSTMin\ is also NP-Hard. Thus,  we have that \CSTMin\ is NP-Complete.\end{proof}

\begin{proof}[Proof of Claim~\ref{OrderingStructure}]
Let $\mathcal{C}(\ST(K_{\pi^+\circ\theta}^+))\setminus \v$ denote the graph obtained if we delete all the nodes (and all entering or exiting edges incident on these nodes) containing the label of $\v$ in $\mathcal{C}(\ST(K_{\pi^+\circ\theta}^+))$.
Note that
$\mathcal{C}(\ST(K_{\pi^+\circ\theta}^+))\setminus \v$ is exactly the same as $\mathcal{C}(\ST(K_{\pi\circ\theta}))$
whose size was invariant over $\pi$ (because $K$ is a 1-dimensional
simplicial complex). Therefore, it suffices to show that the sum of the
number of edges entering or leaving all nodes containing label of $\v$ in
$\mathcal{C}(\ST(K_{\pi^+\circ\theta}^+))$ is at least
the sum of the number of edges entering or leaving all nodes containing the label of 
$\v$ in
$\mathcal{C}(\ST(K_{\rho^+\circ\theta}^+))$. Let $L$ (resp., $R$)
denote the set of all vertices in $G$ whose label under
$\pi\circ\theta$ appear before  (resp., after) $\pi^+(\theta(\v))$. 
Let $\Delta$ be the number of edges between $L$ and $R$ in $G$. Then,
we claim that
$\lvert\mathcal{C}(\ST(K_{\pi^+\circ\theta}^+))\rvert-
\lvert\mathcal{C}(\ST(K_{\rho^+\circ\theta}^+))\rvert=\Delta \ge
0$. To see why the claim is true, let us try to analyze the position
of the nodes containing label of $\v$ in
$\mathcal{C}(\ST(K_{\pi^+\circ\theta}^+))$. Write $L_v$ for the label
of vertex $\v$, i.e., $L_v=\pi^+(|V|+1)$. If there is an edge between two nodes in $R$ then, we
will have a node with label  $L_v$ in layer 2. If there is an edge
between two nodes in $L$ then, we will have a node with label  $L_v$
in layer 4. Finally, if there is an edge between a node in $R$ and a
node in $L$ then, we  will have a node with label $L_v$ in layer 3. By
moving the position of $\v$ to the last position, we are getting rid
of the appearance of nodes with label $L_v$ from layer 2 and
3. Now, the king is either in $L$ or in $R$, and in both cases, we know that there are edges from the king to every other node in $L\cup R$. 
Therefore,  the disappearance of nodes in layer 3, reduces by
$\Delta$, the number of edges because all edges between nodes in  $L$ in layer 2 and node containing label $L_v$ in layer 4, and all edges between nodes in $R$ in layer 2 and node containing label $L_v$ in layer 4 were already existing in $\mathcal{C}(\ST(K_{\pi^+\circ\theta}^+))$.
\end{proof}

\section{Discussion and Conclusion}
In this paper, we introduced a compression technique for the Simplex
Tree without compromising on functionality. Additionally, we have proposed two new
data structures for simplicial complexes -- the Maximal Simplex Tree and
the Simplex Array List. We observed that the
Minimal Simplex Automaton is generally
smaller than the Simplex Automaton. Further, we showed that the Maximal Simplex Tree  is compact and
that the Simplex Array List is efficient (and compact when $d$ is fixed). This is summarized in Table~\ref{tab:OperationsonSAL}.

The transitive closure of MxST may have a node, with as many as $kd$ outgoing edges to
neighbors containing the same label. $\CST$ reduces the number of
outgoing edges to such neighbors with the same label from $kd$ to $d$, making it much more powerful. In short, it reduces the non-determinism of their equivalent automaton representation. Also, most complexes observed in practice have $k$ to be a low degree polynomial in $n$. Example~\ref{NFAvsDFA} and Lemma~\ref{boundOnCompression} both deal with complexes where $k$ is small. Further, all hardness results in section~\ref{S8} are for complexes of dimension at most 2. Thus, complexes where either $k$ or $d$ is small are interesting to study and for these cases, $\CST$ is very efficient.

Marc Glisse and Sivaprasad implemented SAL \cite{MarcSivaprasad} for Data Set 1 on some values of $r$, and then performed insertion and removal of random simplices, and contracted randomly chosen edges. They made the following observations: 
\begin{itemize}
\item In all their experiments ($n\approx 10000,r\in [0.2,0.5]$), size of 1-SAL was significantly smaller than size of ST, and in most cases ST would run out of memory.
\item They found that 1-SAL outperformed 0-SAL in low dimensions. However, 0-SAL performed better than 1-SAL in higher dimensions.
\item They modified 1-SAL by storing fewer edges and this was noted to save a factor of two in size (in dimension 30), over the 1-SAL proposed in this paper.
\end{itemize}
Therefore, it would be worth exploring for which class of simplicial complexes, $i-$SAL is the best data structure in the SAL family (for every $i\in\mathbb{N}$). 

Trie Compression, like that of $\mathcal{M}(\mathrm{SA})$, are efficient techniques when the trie is assumed to be static. However, over the last decade, this has been extended using Dynamic Minimization - the process of maintaining an automaton minimal when insertions or deletions are performed. This has been well studied  in \cite{DynamicMinimization1} and \cite{DynamicMinimization2}, and extended to  acyclic automata in \cite{DynamicMinimization3} which would be of particular interest to us. Interestingly, it appears that in all of the works above, the finiteness of the language plays no special role and, for the specific case of SA, results may be made sharper.

Another direction, is to look at approximate data structures for simplicial complexes, i.e., we store almost all the simplices (introducing an error) and gain efficiency in compression (i.e.,\ little storage).
This is a well explored topic in automata theory called hyperminimization \cite{Hyperminimization} and since our language is finite, $k-$minimization~\cite{kMin} and cover automata~\cite{cover} might give efficient approximate data structures by hyperminimizing SA. We motivate this with the help of building complexes from a random sample of a large data set. By sampling we are bound to lose information and, instead of taking a random sampling we can look at constructing hyperminimized SA over the entire data set dynamically. It will be interesting to know if the power of randomness can overcome `smart' approximations.

Theorem~\ref{MSALabel} can be generalized to give the following hardness result: Given a word $w$ on alphabet set $\Sigma$ and lexicographic ordering $\theta$ on $\Sigma$, let $\mathcal{L}_\theta$ be an operation which removes all duplicate letters in the word and rearranges the letters of the word in lexicographic order given by $\theta$. Given a language $L$ on alphabet set $\Sigma$, we define $\mathcal{L}_\theta(L)=\{\mathcal{L}_\theta(w)\mid w\in L\}$. Therefore with these definitions, we give the following general result:
\begin{theorem}\label{AutomataLabel}
Given a finite language $L$ (explicitly through the words in $L$) on alphabet set $\Sigma$ and an integer $x$, it is NP-Hard to decide if there exists an ordering $\theta$ on $\Sigma$ such that the size of the smallest DFA recognizing $\mathcal{L}_\theta(L)$ is less than $x$.
\end{theorem}

Theorem~\ref{CSTLabel} and \ref{AutomataLabel} provide a new dimension to the hardness results obtained by Comer and Sethi in \cite{CS76}. It would be worth exploring this direction further. Also, it would be interesting to find approximation algorithms for \MSAMin. 

Performance of $i$-SAL depends on the value of $\Gamma_i$ - are there any interesting bounds on $\Gamma_i$ for some subset of nice simplicial complexes of $\mathcal{K}(n,k,d,m)$? Finally, proving better bounds on extent of compression remains an open problem and may be geometric constraints will eliminate pathological examples which hinder in proving good bounds on compression.

\section{Acknowledgement}
We would like to thank Eylon Yogev for helping with carrying out some experiments.  We would like to thank Rajesh Chitnis for pointing out a short proof of Lemma~\ref{PoorHard}. We would like to thank Fran\c{c}ois Godi for pointing out a mistake in the  analysis of the cost of the edge contraction operation for the Simplex Array List as it appeared in \cite{BKT15}. We would like to thank Marc Glisse for several comments on earlier versions of this paper and also for pointing out a tightening of the cost of the edge contraction operation for the Simplex Array List. We would like to thank Marc Glisse and Sivaprasad S.\ for implementing SAL and sharing their results with us. Finally, we would like to thank Dorian Mazauric for pointing out Theorem~\ref{MxSTLabel}.

\bibliographystyle{customurlbst/alphaurlpp}

\clearpage

\appendix

\section{Adapted Hopcroft's Algorithm}\label{Hopcroft}
We precisely describe here Hopcroft's algorithm adapted to the compression of  ST to provide as output $\CoST$. The idea is to write the Simplex Tree as the Simplex Automaton: each vertex of the tree becomes a state in the automaton. Then, one just needs to reduce the Simplex Automaton using Hopcroft's algorithm and finally obtain the compressed Simplex Tree by partitioning the states of the Minimal Simplex Automaton. We see in Figure~\ref{fig:CompressedSimplexTreeExample}, the compressed Simplex Tree of the simplicial complex described in Figure~\ref{fig:SimplicialComplexExample}.
\begin{algorithm}[!h]
\caption{Hopcroft's Algorithm for compression of Simplex Tree}
\begin{algorithmic}[1]
\renewcommand{\algorithmicrequire}{\textbf{Input:}}
\renewcommand{\algorithmicensure}{\textbf{Output:}}
\REQUIRE A Simplex Tree $\mathcal{T}$. Let $L$ be the set of leaves,
$N$ be the set of internal nodes and $V$ be the set of the vertices of the simplicial complex.
\ENSURE A compressed Simplex Tree $\mathcal{T}^\prime$.
\STATE $\mathcal{S} \leftarrow \emptyset $
\STATE $\mathcal{P}\leftarrow \{L,N\}$
\FOR {$v\in V$}
\STATE Add $\left(L,v\right)$ in $\mathcal{S}$.
\ENDFOR
\WHILE {$\mathcal{S}\neq\emptyset$}
\STATE Pop one element $(C,v)$ in $\mathcal{S}$.
\FOR {each $B\in\mathcal{P}$ such that there exists $n_1$ and $n_2$ in
  $B$ such that there is an edge in $\mathcal{T}$ from $n_1$ to a node of $C$
  labelled by $v$, and it is not the case for $n_2$ }
\STATE $B^\prime \leftarrow \{n\in B | \textrm{there is an edge from
}n\textrm{ to a node in }C\textrm{ labelled by }v\}$
\STATE $B^{\prime\prime} \leftarrow \{n\in B | \textrm{there is no
  edge from }n\textrm{ to a node in }C\textrm{ labelled by }v\}$
\STATE Replace $B$ by $B^\prime$ and $B^{\prime\prime}$ in $\mathcal{P}$.
\FOR {$w\in V$}
\IF {$(B,w)\in \mathcal{S}$}
\STATE Replace $(B,w)$ by $(B^\prime,w)$ and $(B^{\prime\prime},w)$ in $\mathcal{S}$.
\ELSE
\STATE $D \leftarrow \begin{cases}B^\prime\textrm{ if
  }|B^\prime|\leq|B^{\prime\prime}| \\B^{\prime\prime}\textrm{ otherwise.}\end{cases}$
\STATE Add $(D,w)$ in $\mathcal{S}$.
\ENDIF
\ENDFOR
\ENDFOR
\ENDWHILE
\STATE Now we just have to build the compressed Simplex Tree $\mathcal{T}^\prime$.
\STATE Let $S$ be the part in $\mathcal{P}$ which contains the root of $\mathcal{T}$.
\STATE $\mathcal{T}^\prime \leftarrow$ initial state: $(S,\star)$ labeled $\star$.
\STATE $\mathcal{V}\leftarrow (S,\star)$
\WHILE {$\mathcal{V}\neq\emptyset$}
\STATE Pop an element $(B,w)$ in $\mathcal{V}$.
\FOR {$v\in V$}
\STATE Let $b$ be an element in $B$. \textbackslash* It will be a representative of $B$.
\algstore{hopcroft}
\end{algorithmic}
\end{algorithm}

\begin{algorithm}[!h]
\begin{algorithmic}[1]
\algrestore{hopcroft}
\STATE Let $a$ be (if it exists) the node in $\mathcal{T}$ reached from $b$ by reading $v$.
\STATE Let $A$ be the part in $\mathcal{P}$ which contains $a$.
\IF {$(A,v)$ is not already in $\mathcal{T}^\prime$}
\STATE Add a vertex $(A,v)$ in $\mathcal{T}^\prime$.
\STATE Push $(A,v)$ in $\mathcal{V}$.
\ENDIF
\STATE Add an edge from $(B,w)$ to $(A,v)$ in $\mathcal{T}^\prime$.
\ENDFOR
\ENDWHILE
\end{algorithmic}
\end{algorithm}


\section{Operations on the Maximal Simplex Tree}\label{MxSTOperations}
We provide below a list of basic operations on the MxST. In the following subsections we denote by $T_{\MxST}$ the maximum time taken to search for particular children of any node in MxST ($T_{\MxST}=\mathcal{O}(\log n)$ for red--black trees).
\subsection{Identifying the Maximal Cofaces of a Simplex} 

As seen in section~\ref{S5}, a simplex $\sigma\in K$ is implicitly stored in MxST$(K)$ as a face of its (at most $k$) maximal cofaces. Identifying the maximal cofaces of a simplex means to find the leaves of MxST$(K)$ that contain the maximal cofaces of $\sigma$ and also to find, on each path from the root to such a leaf, the location of the vertices of $\sigma$.

We formalize this by first defining a bijection $f$ from the set of all maximal simplices of $K$ to the set of all leaves in $\MxST(K)$ and also define a function $g$ that maps every simplex $\sigma\in K$ to the set of all maximal simplices in $K$ which contain $\sigma$ (thus if $f(\sigma)=\emptyset$ then $\sigma\notin K$). 
For simplicity, we will denote $f(g(\sigma))$ by \Ls\  and $\mid$\Ls$\mid$ by \ks\ from now on. Since identifying a simplex $\sigma$ reduces to traversing MxST$(K)$, this operation costs $\mathcal{O}(kdT_{\MxST})=\mathcal{O}(kd\log n)$.

\subsection{Insertion}
Let $\sigma$ be a  simplex not in $K$ and let $K'$ be the complex obtained by adding $\sigma$ and its subsimplices to $K$. Observe that $\sigma$ is necessarily maximal in $K'$ and write $d_{\sigma}$ for the dimension of $\sigma$.  We describe how  to insert $\sigma$ in MxST. 
We first  check  if there exists  maximal simplices in MxST$(K)$  that are contained in $\sigma$. This can be done in $\mathcal{O}(kd_{\sigma}T_{\MxST})$ time, by looking at the MxST truncated to depth $d_\sigma$. If such simplices exist,
we will need to delete them before inserting $\sigma$, which takes time at most $\mathcal{O}(kd_{\sigma}T_{\MxST})$ (see analysis of step 3 of Algorithm~\ref{algo3}).
Then, we insert $\sigma$ in MxST. This takes at most $\mathcal{O}(d_{\sigma}T_{\MxST})$ time,  which is significantly better than the time taken for the Simplex Tree, which needs $\mathcal{O}(2^{d_\sigma}T_{\MxST})$ time. 
We conclude that the time for inserting $\sigma$ is $\mathcal{O}(kd_{\sigma}T_{\MxST})=\mathcal{O}(kd_\sigma\log n)$.

\subsection{Removing a Face}
Given a simplex $\sigma$, we have to remove it (and its cofaces) from the MxST. We can perform the operation of removing simplices (the simplex and its cofaces) as described in Algorithm \ref{algo3}.

\begin{algorithm}[!h]
\caption{Removing simplices in Maximal Simplex Tree}\label{algo3}
\begin{algorithmic}[1]
\renewcommand{\algorithmicrequire}{\textbf{Input:}}
\renewcommand{\algorithmicensure}{\textbf{Output:}}
\REQUIRE A Maximal Simplex Tree and a simplex $\sigma$.
\ENSURE A Maximal Simplex Tree.
\STATE Compute \Ls.
\FOR {each $\Gamma\in$ \Ls}
\STATE Remove the branch ending at $\Gamma$.
\STATE for every vertex $v$ in $\sigma$ insert $\Gamma\setminus\{v\}$.
\ENDFOR
\end{algorithmic}
\end{algorithm}

Step 1 was shown earlier to take $\mathcal{O}(kdT_{\MxST})$ time. Since $\Gamma$ is maximal, it is associated to a leaf in the tree. Let $P(\Gamma)$ be the path from the root to the leaf associated to $\sigma$.  For removing a maximal simplex $\Gamma$, one needs to locate on $P(\Gamma)$, the last node $w$ such that the out--degree for the edges is strictly more than one (it corresponds to the last node which is shared with another maximal simplex). Then, one just has to delete the edge on the corresponding path going from $w$. So removing one maximal simplex (at step 3) takes time $\mathcal{O}(dT_{\MxST})$ (since we might have to potentially delete up to $d$ nodes). At step 4, $d_{\sigma}$ facets of $\Gamma$ need to be inserted and each insertion was shown earlier to be doable in time $\mathcal{O}(T_{\MxST}d_{\Gamma})$) (since we do not have to check  if 	there exists maximal simplices in MxST$(K)$  that are contained in facets $\Gamma$). The total algorithm costs time $\mathcal{O}($\ks$ dT_{\MxST}+ $\ks$ T_{\MxST}d_{\Gamma} ) = \mathcal{O}(kd\log(n))$.


%

\subsection{Elementary Collapse}
 An elementary collapse consists of removing  both simplices of a free pair. It preserves the homotopy type of the complex. We break down the computation of an elementary collapse into 3 steps which is described in Algorithm~\ref{algo4}.
 
\begin{algorithm}[!h]
\caption{Elementary collapse in Maximal Simplex Tree}\label{algo4}
\begin{algorithmic}[1]
\renewcommand{\algorithmicrequire}{\textbf{Input:}}
\renewcommand{\algorithmicensure}{\textbf{Output:}}
\REQUIRE A Maximal Simplex Tree and a pair of simplices $(\tau,\sigma)$.
\ENSURE A Maximal Simplex Tree.
\STATE Check if $(\tau,\sigma)$ is a free pair.
\STATE Delete $\tau$.
\STATE Insert the facets of $\tau$ which are different from $\sigma$, and are now maximal.
\end{algorithmic}
\end{algorithm}

It easily follows from the above results that Step 1 takes $\mathcal{O}(kd\log(n))$ time. As for the removing operation, step 2 is feasible in time $\mathcal{O}(dT_{\MxST})$ (see analysis of step 3 of Algorithm~\ref{algo3}). For step 3, we have $d_\tau$ facets to check if they are now maximal and then insert the ones that are indeed maximal. This can be done in time $\mathcal{O}(kdd_\sigma\log(n))$. 

\subsection{Edge Contraction}
Edge contraction is another operation on simplicial complexes that preserves the homotopy type under certain conditions and which can be implemented on the MxST using the above operations. 
We break down the computation of an edge contraction into 3 steps  which is described in Algorithm~\ref{algo5}.

\begin{algorithm}[!h]
\caption{Edge contraction in Maximal Simplex Tree}\label{algo5}
\begin{algorithmic}[1]
\renewcommand{\algorithmicrequire}{\textbf{Input:}}
\renewcommand{\algorithmicensure}{\textbf{Output:}}
\REQUIRE A Maximal Simplex Tree and a pair of vertices $(u,v)$.
\ENSURE A Maximal Simplex Tree.
\STATE Replace the label $u$ by the label $v$ each time $u$ appears in the MxST.
\STATE Store the list $L$ of all the maximal simplices (lexicographically).
\STATE Build the MxST from $L$.
\end{algorithmic}
\end{algorithm}

Steps 1 and 3 can be done in time $\mathcal{O}(kd\log n)$ (follows from our previous results). However, to perform step 2, we need to remove the simplices which were previously maximal, but are no longer maximal after the edge contraction. This can be done in time $\mathcal{O}(kd\cdot k^\prime)$ \cite{Y92}, where $k^\prime$ is the number of maximal simplices after the edge contraction. So an edge contraction can be performed in time $\mathcal{O}(kd\log n+kdk^\prime)=\mathcal{O}(kd(k+\log n))$. 
Finally, we note that the above algorithm   is a multiplicative factor of $k/\log n$  in the worst case away from the upperbound for MxST. For any $k,d$, let us consider the simplicial complex in Example~\ref{EdgeOptimal}.
\begin{example}\label{EdgeOptimal}
Let $K\in\mathcal{K}(kd/4 + k/2 + d+1,k,d,m)$ be defined by the union of the sets of maximal simplices given below:
\begin{enumerate}
\item $\{1,2,\dots , d,d+i|1\le i\le k/2\}$.
\item $\{i\cdot d+1+k/2,i\cdot d+2+k/2,\dots , i\cdot d+k/2+d\mid 1\le i\le k/4\}$.
\item $\{i\cdot d+1+k/2,i\cdot d+2+k/2,\dots , i\cdot d+k/2+d-1, kd/4 + k/2 + d+1\mid 1\le i\le k/4\}$.
\item $\{1,kd/4 + k/2 + d+1\}$.
\end{enumerate}
\end{example}
If the vertex $kd/4 + k/2 + d+1$ is contracted to $1$ then, the new simplicial complex is generated by the union of the sets of maximal simplices given below:
\begin{enumerate}
\item $\{1,2,\dots , d,d+i|1\le i\le k/2\}$.
\item $\{i\cdot d+1+k/2,i\cdot d+2+k/2,\dots , i\cdot d+k/2+d\mid 1\le i\le k/4\}$.
\item $\{1, i\cdot d+1+k/2,i\cdot d+2+k/2,\dots , i\cdot d+k/2+d-1\mid 1\le i\le k/4\}$.
\end{enumerate}

The new MxST has at least $k(d-1)/4$ new nodes (of size $\log n$) that need to be added. Thus, we have a lower bound of $\Omega(kd\log n)$  on the operation of edge contraction for MxST.  

\end{document}